\documentclass[11pt]{article}
\newcommand{\be}{\begin{eqnarray}}
\newcommand{\ee}{\end{eqnarray}}
\newcommand{\ba}{\begin{eqnarray*}}
\newcommand{\ea}{\end{eqnarray*}}

\tiny\normalsize

\usepackage{graphicx} 
\usepackage{booktabs} 
\usepackage{amsmath}
\usepackage{mathtools}
\usepackage{color, colortbl}
\usepackage{tikz}
\usepackage{multirow}
\usepackage[english]{isodate}

\usetikzlibrary{fit,positioning}
\usetikzlibrary{backgrounds}

\definecolor{blue}{rgb}{0,0,1}

\newtheorem{theorem}{Theorem}[section]

\newtheorem{lemma}[theorem]{Lemma}
\newtheorem{remark}[theorem]{Remark}
\newtheorem{proposition}[theorem]{Proposition}

\newenvironment{proof}[1][Proof]{\begin{trivlist}
\item[\hskip \labelsep {\bfseries #1}]}{\end{trivlist}}

\newcommand{\qed}{\nobreak \ifvmode \relax \else
      \ifdim\lastskip<1.5em \hskip-\lastskip
      \hskip1.5em plus0em minus0.5em \fi \nobreak
      \vrule height0.75em width0.5em depth0.25em\fi}

\usepackage{float}
\usepackage[toc,page]{appendix}
\usepackage[noadjust]{cite}

\usepackage[colorlinks,citecolor=blue]{hyperref}
\usepackage[inline]{enumitem}
\usepackage[affil-it]{authblk}
\usepackage{physics}

\usepackage{mathtools}

\usepackage{enumitem}
\usepackage{algorithm}
\usepackage[noend]{algpseudocode}

\makeatletter
\def\BState{\State\hskip-\ALG@thistlm}
\makeatother

\graphicspath{
  {../Codes/Figures/}
}

\usepackage{amssymb}
\usepackage[square,sort,comma,numbers]{natbib}
\newcommand{\ep}{\varepsilon}

\newcommand{\E}{\mathbb{E}}

\newcommand{\bitem}{\begin{itemize}}
\newcommand{\eitem}{\end{itemize}}
\newcommand{\benum}{\begin{enumerate}}
\newcommand{\eenum}{\end{enumerate}}
\newcommand{\beq}{\begin{equation}}
\newcommand{\eeq}{\end{equation}}
\newcommand{\beqs}{\begin{equation*}}
\newcommand{\eeqs}{\end{equation*}}

\title{Confidence Intervals for the Number of Components in Factor Analysis and Principal Component Analysis via Subsampling}

\author[1]{Chetkar Jha  \thanks{\texttt{Chetkar.Jha@pennmedicine.upenn.edu}} }
\author[1]{Ian Barnett  \thanks{\texttt{ibarnett@pennmedicine.upenn.edu}} }
\affil[1]{Department of Biostatistics, Epedemiology and Informatics, University of Pennsylvania}

\date{\printdayoff\today} 

\begin{document}
\bibliographystyle{cj}
\maketitle

\begin{abstract}
Factor analysis (FA) and principal component analysis (PCA) are popular statistical methods 
for summarizing and explaining the variability in multivariate datasets. By default, FA and PCA assume the number of components or factors to be known \emph{a priori}. However, in practice the users first estimate the number of factors or components and then perform FA and PCA analyses using the point estimate. Therefore, in practice the users ignore any uncertainty in the point estimate of the number of factors or components. For datasets where the uncertainty in the point estimate is not ignorable, it is prudent to perform FA and PCA analyses for the range of positive integer values in the confidence intervals for the number of factors or components. We address this problem by proposing a subsampling-based data-intensive approach for estimating confidence intervals for the number of components in FA and PCA. We study the coverage probability of the proposed confidence intervals and provide non-asymptotic theoretical guarantees concerning the accuracy of the confidence intervals. 
As a byproduct, we derive the first-order \emph{Edgeworth expansion} for spiked eigenvalues of the sample covariance matrix when 
the data matrix is generated under a factor model. We also demonstrate the usefulness of our approach through numerical simulations and by applying our approach for estimating confidence intervals for the number of factors of the genotyping dataset of the Human Genome Diversity Project.
\end{abstract}



\section{Introduction}
\subsection{Motivation}
Factor analysis (FA) and principal component analysis (PCA) are well-established statistical techniques for
explaining the variability in multivariate datasets. More precisely, for a data matrix $\textbf{X}$ consisting of $p$ features 
collected on $n$ samples, researchers use FA for attributing the dependencies among the $p$ features 
to few latent factors, whereas they use PCA for explaining the variability observed in terms of the first few principal components (see Anderson (1958) pg 569 \cite{anderson1958introduction}). Moreover, FA and PCA are increasingly becoming popular as dimension reduction techniques for large multivariate datasets (see \cite{bai2008factor_review}, \cite{jolliffe2016}, \cite{abraham2014pca}, \cite{tsuyuzaki2020pca_scRNA}). Although FA and PCA continue to remain popular statistical techniques, the concern related to selecting the number of factors or components in FA and PCA remains, see Costello and Osborne (2005) \cite{costello2005}.
For the sake of simplicity, henceforth we use components for denoting both factors and components in FA 
and PCA.

By default, FA and PCA assume the true number of components to be known \emph{a priori}. However, in practice the users
use the point estimate as a substitute for the true number of components and thereby ignore any uncertainty in the point 
estimate. Therefore, for all practical purposes, estimating the number of components is the most important decision 
for effectively performing FA and PCA analyses, (see \cite{costello2005}, \cite{zwick1982}, \cite{zwick1986}, \cite{connor2000}, 
\cite{thompson2004}, \cite{larsen2010},\cite{fan2019factor}).
Zwick and Velicer (1986) \cite{zwick1986} put this in context.
They pointed out that the under-specification of the number of components 
will result in the loss of information of components either by ignoring a component or by combining multiple components.
This loss, due to the under-specification, is rather grave if significant resources were expended to collect a large dataset only to throw
away significant signals during dimension reduction.
On the other hand, the over-specification of the number of components will weaken the estimated signal 
strength by combining the true signal with noise and may result in an incorrect understanding 
of the underlying data generating process. Rather than ignoring the uncertainty in the point estimate, it is 
prudent to quantify the uncertainty in the point estimate using a high confidence interval for the number of components and then perform FA and PCA analyses for all positive integer values in the confidence interval (i.e., the confidence set for the number of components). In this paper, we propose a subsampling-based data-intensive approach for estimating confidence intervals for the
number of components in FA and PCA. 

\subsection{Existing methods for estimating the number of components}\label{existing}The existing methods for estimating the number of components in FA and PCA can be grouped into three categories. The first approach considers the problem of estimating the number of components as a model selection problem. The main idea is that the true model (i.e., the model with the true number of components) will maximize a penalized likelihood function or some other information criteria. Bai and Ng (2002) 
\cite{bai2002factor} proposed an information criterion-based approach for estimating the number of components. 
Later, Bai et al. (2018) \cite{bai2018bic} proved that the methods maximizing Akaike information criteria (AIC) and 
Bayesian information criteria (BIC) are also consistent for estimating the number of components in PCA.

The second approach is motivated by the \emph{scree test} (Cattell (1966) \cite{cattell1966}) method. 
The scree test entails visually inspecting the \emph{scree plot} (i.e., the line plot of ordered eigenvalues of
the sample correlation matrix in the descending order) for an \emph{elbow}. 
The number of components to the left of the elbow in the scree plot
is the estimate of the number of components. Onatski (2009) \cite{onatski2009factor} quantified the 
elbow in the scree plot using the maximum ratio of consecutive eigenvalue gap statistics 
and used it for estimating the number of components.

The third approach estimates the number of components by comparing eigenvalues of the sample covariance 
matrix against an empirical threshold. The number of eigenvalues that are larger than the empirical threshold is
the estimate of the number of components. Horn (1965)'s \cite{horn1965} \emph{parallel analysis} 
prescribed computing the empirical threshold as the mean of eigenvalues of sample correlation 
matrices generated from Monte Carlo (MC) simulations under the null hypothesis. Buja and Eyboglu 
(1992) \cite{buja1992}'s PA estimated the empirical threshold as the mean of the eigenvalues of 
sample correlation matrices of permuted data matrices, where permuted data matrices were
generated by independently permuting column entries of the data matrix. Hong, Sheng, and Dobriban (2020)
\cite{hong2020} extended \emph{parallel analysis} for estimating the number of components in PCA
under heterogenous noise using the \emph{sign flip} method.
Recently, several authors have used the random matrix theory (RMT) for estimating the threshold. 
Dobriban and Owen (2019) \cite{dobriban2019} proposed three different methods, namely: DPA, DDPA, and DDPA+ 
for estimating the number of components that estimated the threshold using the Marchenko-Pastur (MP)
 distribution. Cai, Han, and Pan (2020) \cite{cai2020eigenvalues} estimated the threshold using 
the upper-endpoint of the scale-adjusted Tracy-Widom (TW) high confidence interval.
Ke, Ma, and Lin (2020) \cite{ke2020ci} estimated the threshold using a weighted sum of quantiles of the
Marchenko-Pastur distribution. Fan, Guo, and Zheng (2020) \cite{fan2020factor} estimated the number of components 
by thresholding the unbiased estimators of the largest few eigenvalues of the population correlation matrix
 at the phase-transition threshold of $(1 + \sqrt{\gamma})$, where $\gamma_n \to \gamma$ as $n \to \infty$ and
 the aspect ratio $\gamma_n$ is the ratio of the number $p$ of columns of the data matrix $\textbf{X}$ to the number $n$ of rows of the data matrix $\textbf{X}$.


\subsection{Other related works}

\benum
\item Larsen and Warne \cite{larsen2010} discussed two methods for constructing confidence intervals 
for eigenvalues of the population covariance matrix in exploratory factor analysis. The two methods
use the asymptotic distribution of the eigenvalues in classical low-dimensional asymptotics (i.e., $\gamma_n = p/n \to \gamma= 0$) \citep{anderson1958introduction, anderson1963asymptotic}. As we discuss, these two methods do not 
yield a good approximation of the spiked eigenvalues of the sample covariance matrix (i.e., spiked eigenvalue of the sample covariance matrix is an eigenvalue whose corresponding population eigenvalue is larger than the \emph{phase transition} threshold of $(1 + \sqrt{\gamma}))$ when the underlying data matrix $\textbf{X}$ is generated from a factor model.

\item Recently, Saccenti and Timmerman (2017) \cite{saccenti2017considering} compared \emph{parallel analysis} and the sequential 
TW tests of Johnstone (2001) \cite{johnstone2001distribution} that refit the variance and the number of 
features of the TW distribution to the deflated data. They showed that for the maximum eigenvalue $\lambda_1$
of the sample covariance matrix \emph{parallel analysis} is equivalent to a test based
on the TW distribution. For the further eigenvalues, \emph{parallel analysis}
 can be viewed as relying on the joint distribution of eigenvalues, as opposed to the conditional
 distributions given the previous ones. They observed that under an alternative when the maximum eigenvalue
$\lambda_1$  is small, \emph{parallel analysis} can have a much larger false positive
 rate than the nominal one. We can understand this heuristically as the empirical second largest
 eigenvalue may be larger in this case than under the global null. They argue that the sequential
 TW tests are empirically more accurate, while still being heuristic methods. 

\item Recently, Ke, Ma, and Lin (2020) \cite{ke2020ci} proposed an approach for estimating
confidence intervals for the number of components in FA and PCA using bulk eigenvalue matching analysis.
Their approach rely on smoothing the bulk eigenvalues of the residual covariance matrix using quantiles of the 
Marchenko-Pastur distribution (i.e., eigenvalue matching analysis). In particular, they use the eigenvalue matching analysis and the assumption that the diagonal entries of the residual covariance matrix follow a gamma distribution for estimating confidence intervals for the number of components. As we discuss, we see that the coverage probability of their proposed confidence interval
is not close to the nominal confidence level. We discuss this in the Section \ref{sim}. Moreover, their approach do not provide any theoretical guarantees for the coverage probability of their proposed confidence intervals for finite-sample cases, i.e., non-asymptotic cases.
\eenum

\subsection{Overview}
Let $\textbf{X} = (\textbf{x}_1, \cdots, \textbf{x}_n)^\top$ denote a data matrix comprising of measurements taken on $p$ features for $n$ independent units. Then for the $i^{th}$ unit, the standard linear factor model (see \cite{anderson1958introduction},\cite{dobriban2018dpa}) is given as follows
\be\label{factor}
\textbf{x}_i = \mu + \Lambda\eta_i + \ep_i, 
\ee
where $\mu = (\mu_{1} , \cdots, \mu_{p})^\top$, $\Lambda$ is the
factor loading matrix of dimension $p \times r$, 
$\eta_i = (\eta_{i1}, \cdots, \eta_{ir})^\top$
are the factor scores for the $i^{th}$ sample,
and $\ep_i =(\ep_1,\cdots, \ep_p)^\top$ is the residual vector.
Assume that the data matrix $\textbf{X}$ is centered, then the matrix version of the above
factor model is as follows
\be\label{matrix.factor}
\textbf{X} = Z \Theta^{1/2} \Lambda^\top + \E,
\ee
where $Z$ is the $n \times r$ matrix containing the standardized factor values, 
$\Theta$ is a diagonal matrix with diagonal elements controlling
the scale of variability for each component,
 $\Lambda$ is the loading matrix of dimension $ p \times r$,
and $\E$ is the residual matrix with $\ep_i$ as
the $i^{th}$ row of $\E$. 

Our approach is a subsampling-based data-intensive approach that uses eigenvalues of conditionally independent
submatrices (conditional independent given $\textbf{X}$) to estimate confidence intervals for the number of components.
The detailed approach is as follows. We sample $b$ conditionally independent submatrices (conditional on the data matrix $\textbf{X}$) of the same dimensions that share the aspect ratio with the data matrix $\textbf{X}$, i.e., $\gamma_n = \frac{p}{n}$, where $b$ is a positive integer greater than one. Subsequently, we divide $b$ conditionally independent submatrices into  a set of $(b-1)$ submatrices and the remaining submatrix (henceforth, we will refer to it as the lone submatrix). Assume that the eigenvalues of any Hermitian matrix are always arranged in the decreasing order from the  largest eigenvalue to the smallest eigenvalue. Using the eigenvalues of sample covariance matrices of the $(b-1)$ submatrices, we construct empirical confidence intervals for the (unknown) eigenvalues of the population covariance submatrix at the confidence level $\beta$, where $0.5 < \beta < 1$. Then, our preliminary estimate for the number of components is one less than the first index of the empirical confidence intervals that does not contain corresponding eigenvalue of the sample covariance matrix of the lone submatrix. For example, in the Figure \ref{F:2}, the first index of the eigenvalues of the sample covariance matrix of the lone submatrix that does not lie in the corresponding empirical confidence intervals is second. Therefore, the preliminary estimate is one. The preliminary estimate is comparable to the third approach in Section \ref{existing} in the following manner. Rather than comparing eigenvalues against an empirical threshold (i.e., testing whether eigenvalues lie in a one-sided confidence interval), we use the test whether
eigenvalues lie in two-sided confidence intervals for estimating the number of components.
However, the preliminary estimate is weak because we have discarded a large portion of the
data matrix $\textbf{X}$. For improving the estimate and estimating confidence intervals for the number of components, 
we average the preliminary estimate independently (condtional on the data matrix $\textbf{X}$) across sets of $(b-1)$ submatrices and sets of lone submatrix thus estimating confidence interval for the number of components at the confidence level $\beta$, see Section \ref{methods} for details. For optimizing the performance of our method, we need to appropriately select our hyperparameters, including the choice of $b$ and $\beta$. As we show later, the optimum selection of the hyperparameters $b$ and $\beta$ in 
our approach is connected with the first-order \emph{Edgeworth} expansion of the spiked eigenvalues of
the sample covariance matrix of the lone submatrix.

We briefly discuss our main theoretical result. Consider the setup of spiked PCA, first considered by \cite{johnstone2001distribution}. Suppose the eigenvalues of the population covariance matrix $\Sigma$ of the data matrix $\textbf{X}$ generated under PCA consists of $r$ eigenvalues that are greater than one and $(p-r)$ eigenvalues that are one. For high-dimensional asymptotic regime (i.e., $\gamma_n = p/n \to \gamma \in (0, \infty)$ 
as $n \to \infty$), Baik, Arous, and P\'{e}ch\'{e} (2005) \cite{baik2005phase} proved \emph{phase transition} for eigenvalues of the sample covariance matrix. In particular, they proved that in the high-dimensional asymptotic regime when an eigenvalue $l$ of the population covariance matrix $\Sigma$ is smaller than the phase-transtion threshold $(1 + \sqrt{\gamma})$, the corresponding estimated eigenvalue, after appropriately scaling and centering, converges to the TW distribution with a convergence rate of $n^{2/3}$. Moreover, they proved that in the high-dimensional asymptotic regime when the eigenvalue $l$ is greater than the phase-transition threshold $(1 + \sqrt{\gamma})$, (i.e., a spiked eigenvalue), the corresponding estimated spiked eigenvalue $\hat{l}$, after appropriately centering and scaling, converges in distribution to a normal distribution with a convergence rate of $n^{1/2}$, i.e.,
\ba
n^{1/2} \left(\frac{\hat{l} - \rho(l, \gamma_n)}{\sigma(l, \gamma_n)}\right) \xrightarrow[]{\mathcal{L}} N(0,1),
\ea
where $\mathcal{L}$ denotes converges in distribution, and the centering and scaling parameter has the following analytical form (see \cite{johnstone2018pca,baik2005phase,paul2007asymptotics})
\be\label{parm}
\rho(l, \gamma)  = l + \frac{\gamma l}{(l -1)},\ \sigma^2(l, \gamma) =2 l^2 \left(1 - \frac{\gamma}{(l-1)^2} \right).
\ee

Improving upon this result, Yang and Johnstone \cite{johnstone2018pca} derived a first-order \emph{Edgeworth} correction 
result for the spiked eigenvalues of the sample covariance matrix. Moreover for FA, the above \emph{phase transition} 
also occurs, but the exact analytical form of centering and scaling parameters are difficult to compute in a general case, (see
Theorem 11.11 of \cite{yao2015large} (pg 234-235, Chapter 11)). Recently, Cai, Han, and Pan (2020) showed that, in that case of FA, the $i^{th}$ (spiked) eigenvalue of the sample covariance matrix asymptotically converges to a normal distribution provided the $i^{th}$ eigenvalue $l_i$ of the population covariance matrix $\Sigma$ is larger than the phase transition threshold $(1 + \sqrt{\gamma})$ and $\gamma_n/l_i \to 0$ as $n \to \infty$. 
 
We use a subsampling approach for estimating the centering and scaling parameters.
In particular, we independently sample $M$ (conditional on the data matrix $\textbf{X}$) submatrices of size $\lfloor n/b \rfloor \times \lfloor p/b \rfloor$. Then, we numerically compute $\rho^\star_j$ and $\sigma^\star_j$ as the average and standard deviation of the $j^{th}$ largest eigenvalues of the subsampled covariance matrices (i.e., sample covariance matrices of the subsampled matrices). Then, the centering and scaling parameter is given as follows
\be\label{parm.FA}
\rho_j( \lfloor n/b \rfloor, \gamma_n) = \rho^\star_j , \sigma_j( \lfloor n/b \rfloor, \gamma_n) = \sigma^\star_j.
\ee

As we discuss, we see the centering and scaling parameter in (\ref{parm.FA}) leads to a better density
approximation of spiked eigenvalues of subsampled covariance matrix, see Figure \ref{F:1}. Analytically, we derive a first-order
\emph{Edgeworth expansion} result for the spiked eigenvalues of the sample covariance matrix
of the lone submatrix. The main idea behind the proof is that a substitution covariance matrix
can be written as a sum of two positive semidefinite matrices. We exploit this structure to first solve the
determinant equation and then obtain the first-order \emph{Edgeworth} expansion for the spiked eigenvalues
of the substitution covariance matrix. Thereafter, we use the Delta theorem to extend this result 
to the spiked eigenvalues of the sample covariance matrix.

Our theoretical result guaranteeing the non-asymptotic coverage accuracy of our proposed confidence intervals follow from the following observations. We observe that the preliminary estimate is essentially a convolution of a lone spiked (sample) eigenvalue and the mean of $(b-1)$ conditionally independent spiked (sample) eigenvalues. Also, we have the first-order \emph{Edgeworth} expansion for the spiked eigenvalues of the sample covariance matrix of the lone submatrix. From Theorem 2.2 of Hall (1993) \cite{hall1993}, the mean of $(b-1)$ spiked eigenvalues of sample covariance matrices of $(b-1)$ submatrices 
also has a first-order \emph{Edgeworth} expansion. Then, our proposed confidence interval is obtained by independently averaging over the choice of lone submatrix and a set of $(b-1)$ independent submatrix conditional on $\textbf{X}$. Then, we use the 
\emph{Berry-Esseen} inequality and the convolution of the above-mentioned first-order \emph{Edgeworth
correction} of lone spiked eigenvalue and the mean of $(b-1)$ spiked eigenvalues to get an error bound of the coverage
accuracy of the proposed confidence interval. As a consequence, the hyperparameters are obtained
by minimizing the above error bound. The optimum choice of the hyperparameter $b$ is $\lfloor n^{1/3} \rfloor$, whereas
the optimum hyperparameter $\beta$ is obtained by optimizing the conditional likelihood over a grid of $\beta$ values . 
We also discover that our proposed confidence intervals have greater accuracy for high confidence levels.

\begin{figure}
\centering
\includegraphics[width=5in]{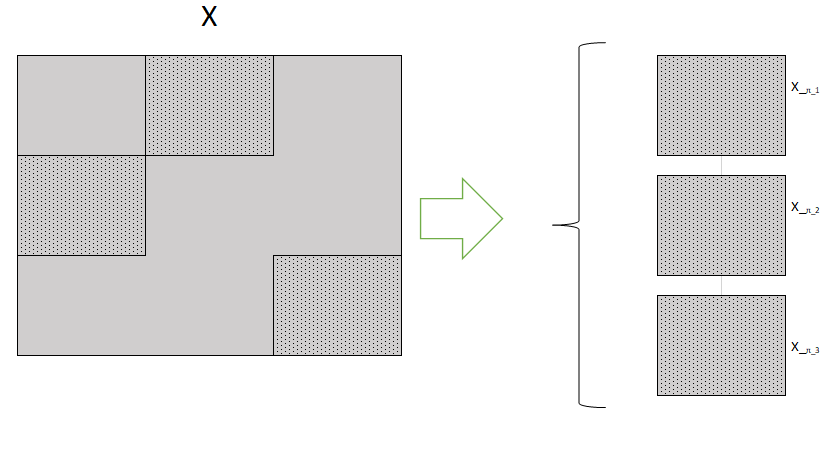}
\caption{Visual representation of the subsampling data matrix. The shaded part on the left denotes the original data matrix $\textbf{X}$ whereas the shaded and dotted part denote submatrices. The submatrices are selected in the manner such that they do not overlap and have the same dimensions while sharing the aspect ratio with the data matrix $\textbf{X}$.}\label{F:0}
\end{figure}

\emph{Organization of the paper}. 
Our contributions are the following: i) We propose a new likelihood-based estimator for estimating the
number of components in FA and PCA. ii) We propose a subsampling-based approach for estimating
confidence intervals for the number of components in FA and PCA. iii) We also derive a non-asymptotic upper bound for 
the coverage accuracy for our proposed confidence intervals. iv) As a byproduct, we
also derive a first-order \emph{Edgeworth} expansion for the spiked eigenvalues of the sample covariance
matrix generated under a factor model. v) Furthermore, we obtain a sharp upper bound on the spectral norm of a cross-product
of two random matrices. The rest of the paper is organized as follows. In Section \ref{methods}, we describe the set up of our problem.
Moreover in Section \ref{methods}, we propose a subsampling algorithm for estimating confidence intervals for the number of components while also proposing a likelihood-based approach for 
estimating the number of components.  In Section \ref{theory}, we derive a \emph{first-order} Edgeworth correction result for the spiked eigenvalues of the sample covariance matrix when the underlying data matrix was generated under a factor model. 
In addition, we derive a non-asymptotic bound for the coverage accuracy for our proposed
confidence intervals of the number of components.
In Section \ref{sim}, we perform numerical evaluation of our approach through various simulation scenarios.
In Section \ref{data}, we demonstrate the usefulness of our approach for carrying out FA and PCA on the genotyping data of Human Genome Diversity Project. Section \ref{disc} ends with the discussion of our proposed approach.
\section{Methods}\label{methods}

\subsection{Notations and Preliminaries}
For any two sequences $\{a_n\}$ and $\{b_n\}$, we say $b_n = O(a_n)$ when there exist 
constants $C$ and $n_0$ such that $|b_n| \le C a_n$ for all $n \ge n_0$.
Likewise, we denote $b_n = o(a_n)$ when for any arbitrarily small $\epsilon > 0$ there exist $n_0$ such that
$|b_n| < \epsilon a_n$ for all $n \ge n_0$. For the rest of the discussion, we assume that
the eigenvalues of any Hermitian matrix are always ordered in the decreasing order
(i.e., arranged from the largest to the smallest). We also denote the greatest integer
function by $\lfloor \cdot \rfloor$. 

The matrix factor model in (\ref{matrix.factor}) is overparametrized and therefore not identifiable. To make the model identifiable, we need to impose condition \ref{a1} on the loading matrix $\Lambda$, see \cite{bai2012factor} for review of identifiability conditions in FA. Moreover, for estimating the model parameters of FA in (\ref{matrix.factor}), we assume that the signal and the error matrices
are independent, see assumption \ref{a2}. Further, we assume that the signal matrix and the error matrix are centered, see
assumption \ref{a3}.
\begin{enumerate}[label= A.\arabic*]
\item\label{a1} (\emph{Identifiability})  Assume that the factor loading matrix $\Lambda$ satisfies $\Lambda^\top \Lambda = I_r$, where $I_r$ is the identity matrix of dimension $r \times r$.
\item\label{a2} (\emph{Independence}) Assume that the signal matrix $Z \Theta^{1/2} \Lambda^\top$ and the noise matrix $\E$ are independent.
\item\label{a3} (\emph{Centering}) Assume that the expected signal and expected noise are zero matrices, i.e., $E(Z \Lambda^\top) = \textbf{0}_{n \times p}$ and $E(\E) = \textbf{0}_{n \times p}$, where $\textbf{0}_{n\times p}$ denotes
a matrix of dimension $n \times p$ with all of its entries as zero. 
\end{enumerate}
Furthermore, we impose additional assumptions (\ref{a4})-(\ref{a6}) on the matrix factor model in (\ref{matrix.factor}) for obtaining
theoretical guarantees for our proposed confidence intervals.
\begin{enumerate}[label = A.\arabic*]
\setcounter{enumi}{3}
\item\label{a4} (\emph{Score Normality}) Assume that the elements of factor score $Z$ are independent realizations from the standard normal distribution, i.e.,
$N(0,1)$.
\item\label{a5} (\emph{Error Normality}) Assume that the elements of noise matrix $\E$ are independent realizations from the standard normal distribution.
\item\label{a6} (\emph{Signal Strength}) Assume that $\Theta = diag(\theta_1, \cdots, \theta_r)$ the diagonal matrix controlling the scale of variability for each factor consists of $\theta_j$ that are at least of the order $O(p^{1 + \delta} )$ for any $j=1, \cdots, r$ and arbitrarily small $\delta > 0$.
\end{enumerate}


Assumption \ref{a4} ensures that all the elements of factor values matrix $Z$ are independent realizations from the standard normal
distribution. Intuitively, assumption \ref{a4} assumes that the component signals are distorted \emph{uniformly} across
the $n$ sample units. Assumption \ref{a5} assumes that the residuals are independent realizations from the standard
normal distribution. Note that the data matrix $\textbf{X}$ can be appropriately centered and scaled to ensure
that the elements of $\E$ (i.e., error terms) in (\ref{matrix.factor}) are normally distributed with mean 
zero and variance one. 
Assumption \ref{a6} prescribes the minimum order of the scale of variability.
As the scale factor of  every component becomes large, the signal components will overwhelm the noise. A common assumption in this context is to assume that the eigenvalues of $Z \Theta Z^\top$
increases with $p$. Bai and Ng (2003)\cite{bai2003} and Fan, Yuan, and Mincheva (2013) \cite{fan2013} estimated the number
of factors under the assumption that the leading eigenvalues of $Z \Theta Z^\top$ increases with $p$. Onatski (2009) \cite{onatski2009} proposed a sequential hypothesis test for estimating the number
of factors when the leading eigenvalues of $Z \Theta Z^\top$ increased with $p^{\theta}$, where $0 < \theta < 1$.
Recently, Cai, Han, and Pan (2020) \cite{cai2020eigenvalues} obtained the asymptotic distribution of the spiked eigenvalues
of the sample covariance matrix under the assumption that $p/(n l_i ) \to 0$, where $l_i$ denote $i^{th}$ spiked eigenvalue of
the population covariance matrix of $\textbf{X}$. In assumption \ref{a6}, we assume that the scale factor, the diagonal elements
of $\Theta$, increased at the rate of $p^{1 + \delta}$ for any arbitrarily small $\delta >0$, i.e., $O(p^{1+ \delta})$. Equivalently, in our setup, this means that the eigenvalues corresponding to the signal matrix are atleast $O(p^{1 + \delta})$.
The rest of the analysis is carried when $\gamma_n = p/n \to \gamma \in (0,1]$. For $\gamma \in [1, \infty)$, we can proceed further by taking the transpose of the data matrix $\textbf{X}$.

\subsection{Likelihood-based estimator} 
We propose a sequential likelihood procedure for estimating the number of components
in the matrix factor model in (\ref{matrix.factor}). We know from Lee and Schnelli (2016)
\cite{leeTWCov16} that the largest eigenvalue of the residual sample covariance matrix ( sample covariance matrix
computed after removing all signals) follows a TW distribution. We use this result to estimate the
number of components in the matrix factor model in (\ref{matrix.factor}). In particular, our estimate of the number of
components is one less than the first index for which TW density
evaluated at the $(r+1)^{th}$ eigenvalue is greater than an arbitrarily small positive number $\delta_0$. An insight into the above observation follows from noting that for spiked eigenvalues that are well above the \emph{phase transition}
threshold of $(1 + \sqrt{\gamma})$, the TW likelihood would approach zero. 
\begin{algorithm}
\caption{Likelihood Based Estimator for the Number of Spikes}\label{alg0}
\begin{algorithmic}[1]
\State Column-wise center and scale the data matrix $\textbf{X}$.
\State Assume the eigenvalues of  the sample covariance matrix of $\textbf{X}$ is given as 
$\hat{l}_1 > \hat{l}_2 > \cdots > \hat{l}_p$.
\State Initialize $\hat{r}=0$.
\State Choose $\delta_0$ to be an arbitrarily small positive number, say $\delta_0 =0.01$.
\State\label{previous.step}If \emph{Tracy-Widom} density evaluated at $\hat{r} + 1$ eigenvalue is less than $\delta_0$ then stop
and return the number of components as $\hat{r}$. Otherwise, go to Step \ref{next.step}.
\State\label{next.step}Increase $\hat{r}$ by one, i.e., $\hat{r} = \hat{r} + 1$ and go to Step \ref{previous.step}.
\end{algorithmic}
\end{algorithm}	

\begin{theorem} Assume that the correlation matrix $\Sigma$ of data matrix $\textbf{X}$ generated from the model in (\ref{matrix.factor}) has exactly $r$ spiked eigenvalues that are greater than $(1 + \sqrt{\gamma})$, where $r$ is an unknown positive integer and the aspect ratio $\gamma$ is the ratio of the number of columns to the number of rows. Moreover, assume that the data matrix $\textbf{X}$ satisfy assumptions (\ref{a1})-(\ref{a6}). Then, we show that the estimate of the number of components in the Algorithm \ref{alg0}
is consistent, as $n \to \infty$, i.e.
\ba
P(\hat{r} = r) \to 1, \text{ as } n \to \infty.
\ea
\end{theorem}

\begin{proof}
See the Appendix.
\end{proof}

The likelihood-based estimator in Algorithm \ref{alg0} in comparison to the other methods, such as PA, is more
sensitive to the normality assumption in (\ref{a4})-(\ref{a5}). Moreover, the likelihood-based estimator assumes
that the spiked eigenvalues to be large, see assumption (\ref{a6}).
As we discuss, the usefulness of the likelihood-based estimator lies in correctly pegging (or centering) our proposed confidence intervals
for the number of components around a consistent estimator for the number of components. 

\subsection{Subsampling algorithm} 
We propose a subsampling-based approach for estimating confidence intervals for the number of components. Our algorithm is described as follows. Consider a data matrix $\textbf{X}$ generated from the factor model in (\ref{matrix.factor}).
We partition the data matrix $\textbf{X}$ into a lone submatrix and a set of $(b-1)$ submatrices, i.e., $\big\{\textbf{X}_{\star}, \{ \textbf{X}_{\pi_i}\}_{i=1}^{b-1} \big\}$ such that all $b$ submatrices are conditionally independent given $\textbf{X}$ and are of same dimensions with the common aspect ratio as that of $\textbf{X}$, i.e., $\gamma_n = p/n$. The rationale for choosing the submatrices of the same dimensions is to make the singular values of all the partitions comparable.
Since we know from Theorem 11.11 of \cite{yao2015large} (pg 234-235, Chapter 11) that the distribution of eigenvalue of any sample
covariance matrix is parametrized in its aspect ratio, therefore we fix the aspect ratio of the above submatrices as that of the aspect ratio of $\textbf{X}$. We use eigenvalues of $(b-1)$ subsampled covariance matrices, (i.e., $\{ \frac{\textbf{X}_{\pi_i} \textbf{X}_{\pi_i}^\top}{\lfloor n/b \rfloor} \}_{i=1}^{b-1}$) to construct empirical confidence intervals for eigenvalues of the unknown population covariance submatrix corresponding to the subsampled covariance matrices at the confidence level $\beta$, where $0 < \beta < 1$. 
Then, using the eigenvalues of the lone subsampled covariance matrix $\frac{\textbf{X}_{\star}\textbf{X}_{\star}^\top}{\lfloor n/b \rfloor}$, we obtain our preliminary point estimate as one less than the first index of the empirical confidence intervals that does not
contain the corresponding eigenvalues of $\frac{\textbf{X}_{\star}\textbf{X}_{\star}^\top}{\lfloor n/b \rfloor}$, see Figure \ref{F:2}.
This approach can be likened to the third approach in the existing literature (Section \ref{existing}).
Recall that the third approach entails the step of comparing eigenvalues of the sample covariance matrix against an empirical threshold that can be interpreted as testing whether eigenvalues lie in a one-sided confidence interval. In comparison, our approach entails
testing whether eigenvalues of the sample covariance matrix of lone submatrix lie in a two-sided confidence interval.

We trim the confidence intervals by $\epsilon_0/\lfloor n/b \rfloor$ on both ends. For the fixed $\beta$, the trimming of
the confidence interval would tend to decrease the preliminary estimate. The trimming constant is scaled by $\lfloor n/b \rfloor$ for keeping the trimming constant comparable with the eigenvalues of the subsampled covariance matrix. The default choice of $\epsilon_0$ is $0.02$ is rather arbitrary and one could select any arbitrarily small positive number. However, choosing $\epsilon_0$ as a large number is not desirable as it might significantly reduce the preliminary estimate to possibly zero. The trimming of the confidence interval does not affect the theory of our approach because it is equivalent
to selecting a lower value of confidence level $\beta$. However, in practice, trimming gives a slightly improved
performance by helping us effectively choose the optimum $\beta$ by refining the grid values of $\{ \beta^{(i)} \}_{i=1}^g$ as a function of the size of the subsample matrix.

For further improving the preliminary point estimate, we generate $M$ random sets where each set consist of a lone submatrix and a set of $(b-1)$ submatrices of equal dimensions with the common aspect ratio $\gamma_n$ such that all $b$ submatrices are conditionally independent given $\textbf{X}$, i.e., $\{ \textbf{X}^{j}_{\star}, \{ \textbf{X}^j_{\pi_i}\}_{i=1}^{b-1} \}_{j=1}^M$.
Averaging the preliminary point estimate over $M$ sets gives us a much-improved
point estimate, say $\hat{r}^{\beta}_1$. For estimating confidence intervals for the number of components, we sample $K$ sets of $(b-1)$ conditionally independent submatrices for every lone submatrix $\textbf{X}^j_{\star}$ 
such that $K$ sets of $(b-1)$ subsampled matrices and the lone submatrix are conditionally independent given $\textbf{X}$, i.e.,
$\{ \textbf{X}^{j}_{\star}, \{ \textbf{X}^{jk}_{\pi_i}\}_{i=1}^{b-1} \}_{j=1, \cdots, M , k = 1, \cdots, K}$. 
Averaging the preliminary point estimate across $M$ sets for $K$ sets of $(b-1)$ conditionally independent submatrices yield $\{ \hat{r}^{\beta}_j \}_{j=1}^{K}$. By construction $\{ \hat{r}^{\beta}_j \}_{j=1}^K$ are $K$
independent random variables conditional on $\textbf{X}$. Therefore, by applying the normal
theory of confidence intervals on $\{ \hat{r}^{\beta} \}_{j=1}^K$, we construct a two-sided confidence interval for the number of components at the confidence level $(1 -\alpha)$, where $(1-\alpha)$ is the user-specified confidence level. The algorithm is summarised in algorithm \ref{alg1}.

It follows from the above discussion that our proposed confidence intervals depend on hyperparameters, namely: $b$, $\beta$, $K$, and $M$ besides the user-specified confidence level $(1-\alpha)$. For instance, a high value of $b$ means that we have smaller submatrices for accurately estimating the spiked eigenvalues and a low value of $b$ means that we have less number of submatrices for accurately estimating centering and scaling parameters of the normal distribution of the spiked eigenvalues. Similarly, a relatively high (low) value of $\beta$ could lead to an overestimation (underestimation) in the estimate of the number of components. As we discuss later, the choice of $M$ and $K$ also affects the performance of our proposed confidence interval.
The hyperparameters in the above subsampling procedure need to be optimized for the optimum coverage probability of the
confidence interval. This is typically the case in subsampling literature where hyperparameters are optimized for improving
the performance of the subsampling algorithm, see \cite{politis2001asymptoticsubsampling}, \cite{wang2018subsampling}. 
We discover that the optimum value of $b$ is $\lfloor n^{1/3} \rfloor$, but we defer that discussion until we develop the 
relevant theory, see Section \ref{theory}. Likewise, the suitable choice of hyperparameters $K$ and $M$ is also deferred to Section \ref{theory}.

\subsection{Treating $\beta$ as an unknown hyperparameter}
The difficulty with the above subsampling algorithm lies in interpreting $\beta$. We treat $\beta$ as an 
unknown fixed hyperparameter that should be optimally selected in a data-intensive manner to improve the performance of
our proposed confidence intervals. In particular, let $\{ \beta^{(i)} \}_{i=1}^g$ vary 
over a grid of $g$ values between $0.5$ and $1$. Moreover, let $\{ \{ \hat{r}^{(i)}_j \}_{i=1, \cdots, g} \}_{j=1, \cdots, K}$ denote
$K$ sets of conditionally independent random variables (conditional on $\textbf{X}$) corresponding to $g$ grid 
values of $\{ \beta^{i} \}_{i=1}^g$.
Using the central limit theorem, for any fixed $i$, the mean of $\{ \hat{r}^{(i)}_j \}_{j=1}^K$ (conditional on $\textbf{X}$)
would be approximately normal. We use this insight to find the optimum value of $\beta$. Assume that $M$
and $K$ are moderately large and fixed.
Then, the optimum $\beta$ is the one that centers the realizations $\{ \hat{r}^{(i)}_j \}_{j=1}^K$ 
around one more than the estimate of the number of components, $r_0 + 1$.
The rationale for centering $\{ \hat{r}^{(i)}_j \}_{j=1}^K$ around $r_0 + 1$ and not $r_0$ is to facilitate our algorithm for computing
confidence intervals for all non-negative integer values of the number of components $r$ including $r$=0.
Let $r_0$ be the estimate of $r$ from the likelihood approach or any other consistent estimator of $r$. Then, we select the optimum $\beta$ that centers the realizations $\{ \hat{r}^{(i)}_j \}_{j=1}^K$ around $r_0 + 1$ while maximizing the following likelihood.
\be\label{opt}
 t\left( \frac{\bar{r}^{(l)} - r_0 - 1 }{ \omega^{(l)} } ; 1, 0\right),
\ee
where $\bar{r}^{(l)} = \sum_{j=1}^K\frac{\hat{r}^{(l)}_j}{K}$, $\omega^{(l)}$ is the standard deviation
of $\{ \hat{r}^{(l)}_j \}_{j=1}^K$, and $t(\cdot; 1, 0)$ denotes  Student's t-distribution with mean zero
and one degree of freedom. 

Using the optimum $\beta^{(l)}$ and $\{ \hat{r}^{(l)}_j\}_{j=1}^K$, we can construct a $(1 - \alpha)\%$ confidence interval
for one more than the true number of components. Substracting one from the lower and upper bound gives us the $(1 - \alpha)\%$
confidence interval for the true number of components, see algorithm \ref{alg1}.



\begin{algorithm}
\caption{Subsampling Algorithm}\label{alg1}
\begin{algorithmic}[1]
\State Input: A matrix $\textbf{X}$ with $n$ rows and $p$ columns, and user-specified significance level $\alpha > 0$.
\State Let the default value of $\epsilon_0 =0.02$.
\State Compute $r_0$ using algorithm \ref{alg0} or using the \emph{deterministic parallel algorithm} proposed by \cite{dobriban2018dpa}.
\State Consider a grid of confidence levels in $[0.5, 1]$ as $\{ \beta^{(l)}\}_{l=1}^g$.
\State Let $b=\lfloor n^{1/3} \rfloor$, and then compute $p_{\star} = \lfloor p/b \rfloor$, $n_{\star} = \lfloor n/b \rfloor$.
\State Randomly select a submatrix $\textbf{X}_{\star}$ of $\textbf{X}$ with dimension $n_{\star}$ rows and  $p_{\star}$ columns.
\State Compute the SVD of  sample covariance matrix $\frac{ \textbf{X}_{\star} \textbf{X}_{\star}^\top }{n_{\star}} = U D_{\star} V^\top$.
\State Store sample	 singular values of $D_{\star}$ in $l^{\star}_1 \ge \cdots \ge l^{\star}_{p_{\star}}$. 
\BState Compute $b-1$ conditionally independent submatrices of equal dimensions shown as follows
\begin{enumerate}[label = \roman*]
\State Denote $(b-1)$ independent subsetted data matrix by $\{ \textbf{X}_{\pi_i} \}_{i=1}^{b-1}$, which are also independent of $\textbf{X}_{\star}$. Let the corresponding $(b-1)$  sample covariance matrices be given by $\left\{ \frac{ \textbf{X}_{\pi_i} \textbf{X}_{\pi_i}^\top  }{n_{\star}} \right\}_{i=1}^{b-1}$,
\State For each $i=1, \cdots, b-1$, compute the SVD of $\frac{ \textbf{X}_{\pi_i} \textbf{X}_{\pi_i}^\top }{n_{\star}} = R_i Q_i S^\top_i$,
\State Store singular values of $Q_i$ in $l^{\pi_i}_1 \ge \cdots \ge l^{\pi_i}_j \ge \cdots l^{\pi_i}_{p_{\star}}$, where $i=1, \cdots, b-1$.
\end{enumerate}
\State Compute $\mu_{j} = \frac{\sum_{i=1}^{b-1} l^{\pi_i}_j}{b-1}$ and $\sigma_j = \sqrt{ \frac{\sum_{i=1}^{b-1}( l^{\pi_i}_j - \mu_j)^2}{b-1}}$.
\BState\label{step2} For $i=1:g$,
\begin{enumerate}[label =\roman*]
\State At $\beta^{(i)}$, define confidence interval $CI^{(i)}_{j}$ as $[\mu_j +z_{\beta^{(i)}/2} \sigma_j + \epsilon_0/n_{\star} , \mu_j + z_{1 - \beta^{(i)}/2} \sigma_j -\epsilon_0/n_{\star}]$,
\State Accept $r^{(i)} =
\max_{ \{j : j \in \{1, \cdots,  p_{\star} \} \}} I\{ \{ l^\star_{s} \in CI^{(i)}_{s} \text{ for all } s \in \{1, \cdots, j \} \} \cap
\{ l^{\star}_{j+1} \notin CI^{(i)}_{j} \} \} $,
\end{enumerate}
\State Repeat Steps (2)-(12) $M$ times to get $M$ independent realizations of $\{ r^{(i)}\}$.
\State Repeat Steps (5)-(13) $K$ times and denote the estimate of the number of spikes by $\{ r^{(i)}_{k,i}\}_{k=1, \cdots, K , i=1, \cdots, M}$.
\State Compute $\hat{r}^{(i)}_k = \frac{\sum_{j=1}^M  r^{(i)}_{k,j}}{M}, \hat{r}^{(i)} = \frac{\sum_{k=1}^K \hat{r}^{(i)}_k}{K-1}$, and $\hat{\sigma}^{(i)}_{r} = \sqrt{\frac{\sum_{k=1}^K(\hat{r}^{(i)}_k - \hat{r}^{(i)})^2}{K}}$.
\State Select $i^{\star}$ as $i$ maximizing the following objective function
\be
i^{\star} = \max_{i : i \in \{1, \cdots, g \}}  t( \frac{\bar{r}^{(i)} - r_0 - 1 }{ \omega^{(i)} }),
\ee
where $\bar{r}^{(i)} = \sum_{j=1}^K\frac{\hat{r}^{(i)}_j}{K}$, $\omega^{(i)}$ is the standard deviation
of $\{ \hat{r}^{(i)}_j \}_{j=1}^K$, $t(\cdot; 1, 0)$ denotes Student's t- distribution with one degree of freedom and mean zero.
\State Using the normal theory of confidence intervals and adjusting for a bias of $one$, the two sided $(1- \alpha)$ confidence interval for $r$ is given as
$[\lfloor \hat{r}^{(i\star)} + z_{\alpha/2} \hat{\sigma}^{i^\star}_r \rfloor - 1,  \lfloor \hat{r}^{(i\star)} + z_{1-\alpha/2} \hat{\sigma}^{i^\star}_r \rfloor - 1]$, where $\lfloor \cdot \rfloor$ is the greatest integer function.
\end{algorithmic} 
\end{algorithm}

\section{Theoretical results}\label{theory}
The performance of a confidence interval is evaluated by the closeness of 
the coverage probability (i.e., for a large number of independent experiments, the percentage of times
the computed confidence interval covers the true parameter) of a confidence interval to its specified confidence level. The closer is the coverage probability to its specified confidence level, the better is the performance of the confidence interval
at that specified confidence level. For providing theoretical guarantee for our proposed confidence intervals, 
we derive a non-asymptotic upper bound for the absolute deviation of the coverage probability from the specified confidence level. 
Using this bound we also establish that the coverage probability of our proposed confidence intervals is consistent. The above theoretical guarantee rely on deriving the first-order \emph{Edgeworth correction} of spiked eigenvalues of the sample covariance matrix when the underlying data was generated from a factor model in (\ref{matrix.factor}).

\subsection{First-order Edgeworth correction of the spiked eigenvalues of substituted covariance matrix}
Recently, Yang and Johnstone (2018) \cite{yang2018edgeworth} proposed a first-order \emph{Edgeworth
correction} for the largest  eigenvalue of the sample covariance matrix when the underlying data matrix $\textbf{X}$
was generated under a principal component model. Unfortunately, their approach cannot be automatically generalized to 
the largest few eigenvalues of the sample covariance matrix generated under a factor model in (\ref{matrix.factor}) for the following reasons. First, their approach cannot be used for solving the determinant equation of eigenvalues of the sample covariance matrix generated under a factor model. Second, since the loadings in FA are distributed across $p$ factors, it is difficult to estimate the analytical form of the population mean and population standard deviation of spiked eigenvalues of the sample covariance matrix, see (\ref{parm}). 

We overcome the above two problems as follows. For the first problem,
we solve the determinant-based equation for a substitution covariance matrix which is positive semidefinite.
Then, borrowing ideas from Yang and Johnstone (2018) \cite{yang2018edgeworth}, we derive a first-order \emph{Edgeworth} correction for the spiked eigenvalues of the substitution covariance matrix. Subsequently, we use the Delta theorem and the upper bound of the spectral norm of the cross product of two random matrices to get a first-order \emph{Edgeworth correction} for the spiked eigenvalues of the sample covariance matrix. For the second problem, we use a subsampling strategy for estimating the centering $\{ \rho_i \}_{i=1}^r$ and scaling parameters $\{ \sigma_i \}_{i=1}^r$ for top $r$ spiked eigenvalues of the sample covariance matrix. Essentially, we repeatedly subsample a submatrix with the same aspect ratio $\gamma_n$ (as that for data matrix $\textbf{X}$) conditional on the data matrix $\textbf{X}$ for empirically estimating the centering $\{ \rho_i\}_{i=1}^r$ and the scaling $\{ \sigma_i\}_{i=1}^r$ parameters for $r$ spiked eigenvalues. Theoretically speaking, the centering and the scaling parameter also depend on the empirical distribution of the $p$ sample eigenvalues of the residual sample covariance matrix $n^{-1} \E \E^\top$.
Assume that the $F_{n}$ denotes the empirical eigenvalue distribution of the $p$ sample eigenvalues of $n^{-1} \E^\top
\E$. As $n$ grows $F_{n}$ converges to the Marchenko-Pastur distribution $F_{\gamma}$ supported
on $[a(\gamma), b(\gamma)]$ when $\gamma \le 1$ and with a point mass of $(1 - \gamma^{-1})$ at zero when $\gamma > 1$,
where $a(\gamma) = (1 - \sqrt{\gamma})^2, b(\gamma) = (1 + \sqrt{\gamma})^2$. The \emph{companion}
empirical distribution of the $n$ eigenvalues of $n^{-1} \E \E^\top$ converges to the companion MP law
$F_{\gamma} = (1- \gamma) I_{[0, \infty)} + \gamma F_{\gamma}.$ In fact, as shown in \cite{johnstone2018pca}
the skewness component is given in terms of the above-defined $F_{\gamma}(\cdot)$.

\begin{lemma}\label{edgeworth.sub}
Assume that the data matrix $\textbf{X}$ generated under the factor model in (\ref{matrix.factor}) satisfy assumptions (\ref{a1})-(\ref{a6}). Let the singular value decomposition of the population covariance matrix $\Sigma$ of the data matrix $\textbf{X}$ be given by $V M V^\top$, where the diagonal matrix $M$ is given as below
\be
M = \begin{bmatrix} L & \textbf{0}^\top_{r \times n-r} \\
\textbf{0}_{r \times n-r} &  I_{n-r \times n-r}
\end{bmatrix},
\ee
 with $L = diag(l_1, \cdots, l_r)$, $l_1 > \cdots > l_r > (1 + \sqrt{\gamma})$, $\textbf{0}_{r\times n-r}$ denoting the matrix of dimension $r\times (n-r)$ with all zero elements, $V$ denoting a unitary matrix of dimension $n \times n$, and $\gamma_n = p/n \to \gamma$.

For $H, \Theta, \E$ defined in (\ref{matrix.factor}), define the substitution covariance
matrix $\Sigma_s$ as below
\be\label{sub.Sigma}
\Sigma_s = \frac{H \Theta H^\top}{n} + \frac{\E \E^\top}{n}.
\ee 

Let $\hat{l}^s_i$ be the $i^{th}$ largest eigenvalue of the substitution covariance matrix $\Sigma_s$ given in (\ref{sub.Sigma}).
Let $\rho_i, \sigma_i$ be the centering and scaling parameter given in (\ref{parm.FA}).
Then, for any $i \in \{1, \cdots, r \}$, the cumulative distribution function of $\mathcal{T}^{s}_{n,i} = n^{1/2}( \hat{l}^s_i - \rho_i )/\sigma_i$  has the following first-order Edgeworth expansion 
\be
P( \mathcal{T}^{s}_{n,i} \le x) = \Phi(x) + n^{-1/2} p_{1,n}(x) \phi(x) + o(n^{-1/2}), \text{ valid uniformly in x},
\ee
where $Y_n =o(1)$ denotes that $\lim_{n\to \infty} P(|Y_n| \ge \epsilon) = 0$, and 

\be
p_{1,n}(x) = \frac{1}{6} \kappa^{-3/2}_{2,n} \kappa_{3,n} (1- x^2) - \kappa^{-1/2}_{2,n} \mu(g_{(n,j)}),
\ee
for $g_{(n,j)}(\lambda) = (\rho_j - \lambda)^{-1}$, $\kappa_{2,n} = 2 F_{\gamma_n}(g_{(n,j)}^2)$, $\kappa_{3,n} = 8 F_{\gamma_n}(g_{(n,j)}^3)$, $\mu(\cdot)$ is the asymptotic mean in the central limit theorem of Bai-Silverstein limit \citep{bai2004clt}.
\end{lemma}


The above lemma exploits the positive semidefinite structure of the substituted covariance matrix to obtain
a first-order \emph{Edgeworth correction} of spiked eigenvalues of the substituted covariance
matrix. However, using the spiked eigenvalues of the substitution covariance matrix in Algorithm \ref{alg1}
would be computationally cumbersome as it would entail estimating $H$ and $\E$ for every submatrix. 
Since Algorithm \ref{alg1} is given in the terms of the sample covariance matrix, we need to extend the above lemma 
to the spiked eigenvalues of the sample covariance matrix. 
For extending the above \emph{Edgeworth correction} result to the spiked eigenvalues of sample covariance matrix, we use an upper bound on the spectral norm of the cross product of two random matrices and the Delta theorem.

\subsection{Spectral Norm of Cross Product of Two Random Matrices}
The upper bound of the expectation of the spectral norm is known for random matrices with independent elements, see \cite{bandeira2016norm},\cite{handel2017norm}, \cite{seginer2000norm}, and \cite{latata2004norm}.
 However, none of the existing results give an upper bound for the expectation
of the spectral norm of a cross product of two random matrices. To this end, we propose the following theorem.
 
\begin{proposition}\label{comp.prop}Let $H, \Theta, \Lambda, \E$ be as defined in (\ref{matrix.factor}) and assume that assumptions 
(\ref{a1})-(\ref{a6}) hold, then we have the following bound
\ba
E( || Z \Theta^{1/2} \Lambda^\top \E^\top + \E \Lambda \Theta^{1/2} Z^\top||) \le (8^22^{1/4} \exp(2) log(2n)) \sqrt{n log(p) log(r)} \max_{1 \le l \le r } \sqrt{\theta}_l .
\ea
\end{proposition} 

\begin{proof}
See the Appendix.
\end{proof}

The proof of the above proposition uses the cancellation of terms in the expansion of trace to derive the above bound of the
spectral norm, see \cite{furedi1981norm},\cite{vu2007norm}, \cite{seginer2000norm} for reference. A comparable approach is to condition the cross product of two matrices on one of the two matrices and then repeatedly apply Theorem 1.1 of \cite{vershynin2011norm} to get the upper bound. Unfortunately, such an approach in our context yields an upper bound of $O(\max_{1 \le l \le r}\sqrt{\theta}_l n)$ which is not sharp compared to Proposition \ref{comp.prop}.
The above spectral upper bound of the cross product of two random matrices along with the Delta theorem allows us
to extend  the \emph{Edgeworth correction} result for the spiked eigenvalues of the substituted covariance matrix 
to the spiked eigenvalues of the sample covariance matrix.

\begin{theorem}\label{final.thm}
Assume the data matrix $\textbf{X}$, generated under the factor
model in (\ref{matrix.factor}), satisfy assumptions (\ref{a1})-(\ref{a6}). Let the singular value
decomposition of the population covariance matrix $\Sigma$ be given by $V M V^\top$, where the diagonal matrix $M$ is given as follows
\be
M = \begin{bmatrix} L & \textbf{0}^\top_{r \times n-r} \\
\textbf{0}_{r \times n-r} &  I_{n-r \times n-r},
\end{bmatrix}
\ee
 with $L = diag(l_1, \cdots, l_r)$, $l_1 > \cdots > l_r > (1 + \sqrt{\gamma})$, $\textbf{0}_{r\times n-r}$ denoting the matrix of dimension $r\times n-r$ with all zero elements, and $V$ denoting a unitary matrix of dimension $n \times n$.

Let $\hat{l}_i$ be the $i^{th}$ largest eigenvalue of the sample covariance matrix. Assume that the centering and scaling parameters 
$\rho_i$ and $\sigma_i$ are given in (\ref{parm.FA}).
Then, we show that the cumulative distribution function of $\mathcal{T}_{n,i} = n^{1/2}( \hat{l}_i - \rho_i )/\sigma_i$ has a first-order Edgeworth expansion as follows
\be
P( \mathcal{T}_{n,i} \le x) = \Phi(x) + n^{-1/2} p_1(x) \phi(x) + o(n^{-1/2}),
\ee
valid uniformly in $x\in \mathbb{R}$, and with
\be
p_1(x) = \frac{1}{6} \kappa^{-3/2}_{2,n} \kappa_{3,n} (1 - x^2) - \kappa^{-1/2}_{2,n}\mu(g_{(n,j)}),
\ee
for $g_{(n,j)}(\lambda) = (\rho_j - \lambda)^{-1}$, $\kappa_{2,n} = 2 F_{\gamma_n}(g_{(n,j)}^2)$, $\kappa_{3,n} = 8 F_{\gamma_n}(g_{(n,j)}^3)$ , $\mu(\cdot)$ is the asymptotic mean in the central limit theorem of Bai-Silverstein limit, see Bai and Silverstein (2004) \cite{bai2004clt}.
\end{theorem}

\begin{proof}
See the Appendix.
\end{proof}
Although, the above theorem is given for the spiked eigenvalues of the sample covariance matrix but in practice we can only compute the centering and scaling parameter for the spiked eigenvalues of the subsampled covariance matrices.
Later in the paper, we use the above theorem to obtain a normal approximation of the spiked eigenvalues of the sample covariance matrix of the lone submatrix.

\subsection{Comparison with the existing approaches}
We numerically compare the above normal approximation with the approximation proposed by Larsen and Warne (2010) \cite{larsen2010} and the first-order \emph{Edgeworth correction} for spiked eigenvalues of PCA, see \cite{johnstone2018pca}.
For finite sample, the mean of the above two approximation are far from the mean obtained using the subsampling approach. This is because the loadings in the FA are distributed across $p$ factors which may cause the population mean to become comparatively bigger or smaller. For a fair comparison, we center the two competing approximations the mean obtained using the subsampling approach. Figure \ref{F:1} suggests that the two competing approximations overestimate the variance of the spiked eigenvalues.

Note that the assumption (\ref{a6}) for FA is far from benign. However, given that we lack an \emph{Edgeworth expansion} result for
the spiked eigenvalues of the sample covariance matrix, we believe that the above result could be useful. Also, recall that Yang and Johnstone (2018) \cite{johnstone2018pca} for the spiked eigenvalues of the sample covariance matrix for PCA under relatively simpler
assumptions. Using these two results, we show that it is possible to gauge the coverage accuracy of our proposed confidence interval
in algorithm \ref{alg1}.

\begin{figure}
\includegraphics[width=5.5in]{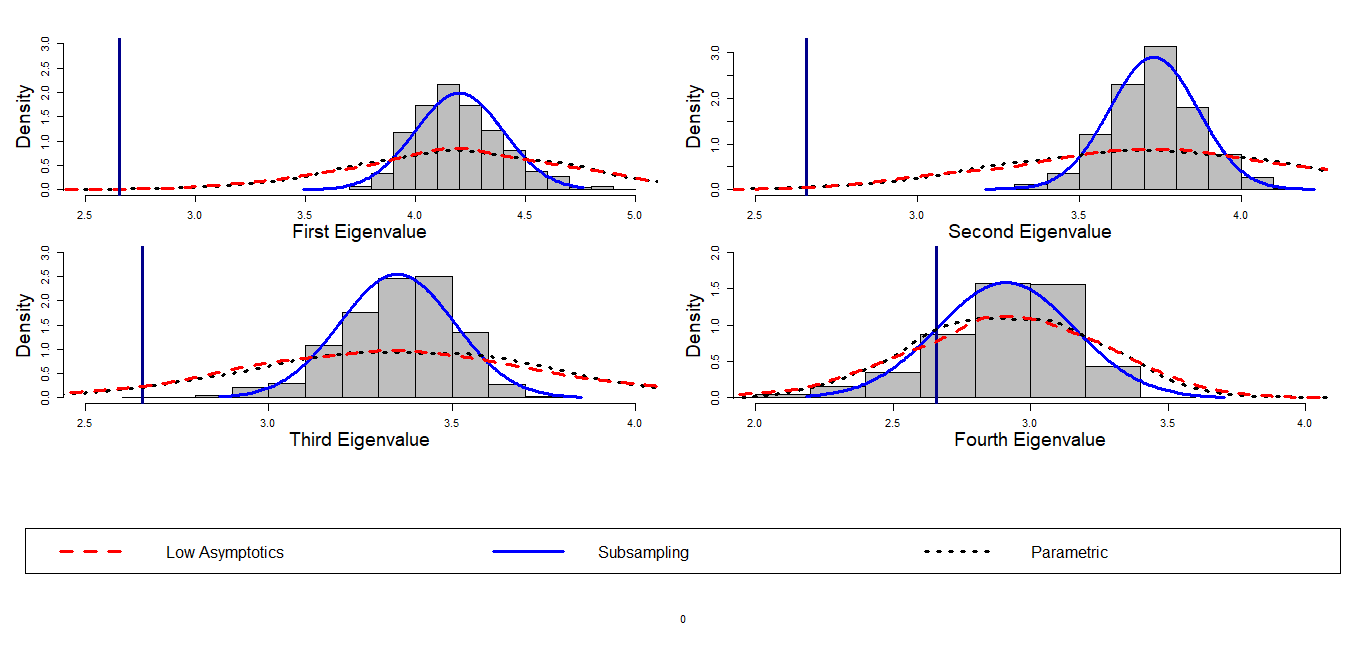}
\caption{Approximation comparison. We generated $1000$ data matrices under FA when $n = 1500$, $p =600$, $r=4$, and the scale factor, i.e. $(\theta_1, \theta_2, \theta_3, \theta_4)=(15,15,15,15)$. We used $1000$ data matrices to plot the histogram of top four spiked eigenvalues of the subsampled covariance matrices. Subsequently, we approximated the histogram of top four spiked eigenvalues using Larsen and Warne (2010)'s \cite{larsen2010} approach (Low Asymptotics),  first-order \emph{Edgeworth correction} for PCA proposed by Yang and Johnstone (2018) \cite{johnstone2018pca} (Parametric), and the subsampling approach. The three approximations along with the histogram is plotted for first eigenvalue, second eigenvalue, fourth eigenvalue and third eigenvalue of the subsampled covariance matrices in clockwise manner. Verical dark blue lines are drawn at $(1 + \sqrt(\gamma))^2$.
}\label{F:1}
\end{figure}


\subsection{Coverage Probability}
At the heart of Algorithm \ref{alg1} is the convolution of two approximately normal random variables, i.e., $\sqrt{b-1}(\frac{ \sum_{j=1}^b\hat{l}^{\pi_j}_i}{b-1} - \rho_i)/ \sigma_i$ and $(l^\star_i - \rho_i)$. 
Recall that the Theorem \ref{edgeworth.sub} states that
$(l^\star_i - \rho_i)$ has an \emph{Edgeworth expansion} whereas Theorem 2.2 of Hall (1993) \cite{hall1993} implies
that $\sqrt{b-1}(\frac{ \sum_{j=1}^b\hat{l}^{\pi_j}_i}{b-1} - \rho_i)/\hat{\sigma}_i$ has an \emph{Edgeworth expansion}.
For the sake of completeness, we restate the above theorem in our setting.

\begin{theorem}
Assume that $\textbf{X}$ generated from (\ref{matrix.factor}). Let  $\{ \textbf{X}_{\pi_i} \}_{i=1}^{b-1}$ be conditionally independent (conditional on $\textbf{X}$) submatrices defined in algorithm \ref{alg1}. Let $\{ \hat{l}^{\pi_j}_i \}_{j=1}^{b-1}$ be $b-1$ $i^{th}$ largest eigenvalue of the sample covariance matrices of $\{ \frac{\textbf{X}_{\pi_i} \textbf{X}^\top_{\pi_i}}{\lfloor n/b \rfloor} \}_{i=1}^{b-1}$. Let $\rho_i$ and $\sigma_i$ be the centering and the scaling parameter defined in (\ref{parm.FA}). 
Then, we show that the cumulative distribution function of $\sqrt{b-1}( \frac{ \sum_{j=1}^b\hat{l}^{\pi_j}_i}{b-1}  - \rho_i)/\sigma_i$  has a first-order Edgeworth expansion as shown below
\be
\label{mean.t} P( \sqrt{b-1}( \frac{ \sum_{j=1}^b\hat{l}^{\pi_j}_i}{b-1}  - \rho_i)/\sigma_i \le x \mid \textbf{X}) = \Phi(x) + (b-1)^{-1/2} p_1(x) \phi(x) + o((b-1)^{-1})
\ee
valid uniformly in $x\in \mathbb{R}$, and with
\be
p_1(x) = \frac{1}{6} \kappa^{-3/2}_{2,n} \kappa_{3,n} (1 - x^2) - \kappa^{-1/2}_{2,n}\mu(g_{(n,j)}),
\ee
for $g_{(n,j)}(\lambda) = (\rho_j - \lambda)^{-1}$, $\kappa_{2,n} = 2 F_{\gamma_n}(g_{(n,j)}^2)$, $\kappa_{3,n} = 8 F_{\gamma_n}(g_{(n,j)}^3)$ , $\mu(\cdot)$ is the asymptotic mean in the central limit theorem of Bai-Silverstein limit, see Bai and Silverstein (2004) \cite{bai2004clt}.
\end{theorem}

We use the above two \emph{Edgeworth expansions} to get a normal approximation of the weighted difference of the two
random variables. Moreover, since the proposed confidence interval in algorithm \ref{alg1} is constructed using a set of
conditionally independent random variables therefore we use \emph{Berry-Esseen} inequality to quantify the accuracy of 
the proposed confidence interval.

\begin{theorem}\label{cov.prob.thm}(Coverage Accuracy)
The absolute deviation between the coverage probability of the confidence interval in Algorithm \ref{alg1} and the confidence level 
$(1-\alpha)$ has the following upper bound.
\be\label{conv.prob}
&& |P( r \in [\hat{\mu} + z_{\alpha/2} \hat{\sigma}, \hat{\mu} + z_{1 -\alpha/2} \hat{\sigma}]) - (1-\alpha)|  \nonumber\\
&&\le   \frac{\rho (2 x^2 + 1) \phi(x)}{6 \sigma^3 \sqrt{K}} + o(\max\{\frac{1}{K}, r\sqrt{\frac{b}{n}}, \frac{r}{b-1} \}),
\ee
where $\sigma$ is the unknown population standard deviation of $\{ \hat{r}_{\beta}\}'s$ in algorithm \ref{alg1}.
\end{theorem}

The above theorem bounds the departure of the coverage probability of our proposed confidence interval from the 
expected confidence level. Clearly, minimizing this bound would improve the performance of the confidence interval.
It is easy to see the error bound is minimized as a function of $b$ when $b = \lfloor n^{1/3} \rfloor$. Moreover, the above
error bound will be minimized by selecting a high value of $K$ while minimizing the skewness $\rho$. 

\begin{remark}\label{remark.b}
The choice of $b = \lfloor n^{1/3} \rfloor$ minimizes the coverage probability in the Theorem \ref{cov.prob.thm}
\end{remark}

\begin{remark}
Note that the R.H.S. of (\ref{conv.prob}) decreases to zero as long as $r n^{-1/3}$ goes to zero. Therefore for higher
value of $r$ the performance of our proposed algorithm is worse. Moreover, we see that the first term in R.H.S. of (\ref{conv.prob})
depends on the normal density evaluated at $z_{\alpha}$, i.e. $\phi(z_{\alpha})$. It is easy to see the performance
of our proposed interval is better when $\phi(z_{\alpha})$ is considerably smaller.
\end{remark}

\subsection{Justification of optimum $\beta$} 

Recall from algorithm \ref{alg1} that $\{\beta_l\}_{l=1}^g$ consist of a grid of values ranging from $50\%$ to $100\%$. From the definition of skewness $\rho$, in the R.H.S. of (\ref{conv.prob}), is 
$E( (\hat{r}_{\beta} - r - 1)^3 )$. 
Let $l(r \mid \{ \hat{r}_{\beta}\}, \textbf{X})$ denote the conditional log likelihood given $\textbf{X}$. Let $\dot{l}(r \mid \{ \hat{r}_{\beta}\}, \textbf{X})$, $\ddot{l}(r \mid \{ \hat{r}_{\beta}\}, \textbf{X})$, and $\dddot{l}(r \mid \{ \hat{r}_{\beta}\}, \textbf{X})$ denote the first, second and third derivative of the conditional log likelihood, where the derivative taken with respect to $r_{\beta}$. Then, applying Taylor series expansion for a fixed $\hat{r}_{\beta}$, we get
\ba
\dot{l}(r + 1 \mid \{ \hat{r}_{\beta}\}, \textbf{X})= \dot{l}(\hat{r}_{\beta} \mid \{ \hat{r}_{\beta}\}, \textbf{X}) + (r + 1 - \hat{r}_{\beta}) \ddot{l}(\hat{r}_{\beta} \mid \{ \hat{r}_{\beta}\}, \textbf{X}) + \frac{ (r + 1 - \hat{r}_{\beta})^2}{2} \dddot{l}(\hat{r}_{\beta} \mid \{ \hat{r}_{\beta}\}, \textbf{X}) 
\ea
Recall by definition, for optimum $\beta^\star$, we have $\dot{l}(\hat{r}_{\beta^\star} \mid \{ \hat{r}_{\beta^\star}\}, \textbf{X})=0$. Then, we have
\be\label{cons.r}
(r + 1 - \hat{r}_{\beta^\star}) = \frac{\dot{l}(r + 1 \mid \{ \hat{r}_{\beta^\star}\}, \textbf{X})}{\ddot{l}(\hat{r}_{\beta^\star} \mid \{ \hat{r}_{\beta^\star}\}, \textbf{X}) + 0.5(r + 1 - \hat{r}_{\beta^\star}) \dddot{l}(\hat{r}_{\beta^\star} \mid \{ \hat{r}_{\beta^\star}\}, \textbf{X})}.
\ee

Then, using the Central Limit Theorem, the boundedness of the third derivative of the conditional likelihood in (\ref{cons.r}) to show that $|r + 1 - \hat{r}_{\beta^\star}|$ is upper bounded by a constant. This can be shown by verifying the Cramer-Rao conditions, see \cite{lehmann99}.

For bounding the coverage accuracy in the R.H.S. of (\ref{conv.prob}), we show that the two terms in the R.H.S. of (\ref{conv.prob})
can be minimized. Consider the first term in the R.H.S. of (\ref{conv.prob}) that is given in terms of  $\rho,\sigma, K$, and $\phi(z_{\alpha/2})$, where $\phi(\cdot)$ is the normal density and $z_{\alpha/2}$ is the normal density quantile evaluated at $\alpha/2$. From the preceeding discussion in this section, $\rho$ is bounded when $\beta$ is selected as described in the algorithm \ref{alg1}.
Suppose for fixed large $K, M, \epsilon_0$, we have $\sigma > \sigma_0 > 0$. Then, the first term becomes small for large quantiles $z_{\alpha/2}$ and large $K$. From Remark \ref{remark.b} and constant $r$, the second term in the R.H.S. of (\ref{conv.prob}) is $o(n^{-1/3})$. For non-asymptotic case, the performance of proposed confidence intervals in the algorithm \ref{alg1} is maximized by selecting $b=\lfloor n^{1/3}\rfloor$, $\beta$ the maximum likelihood estimate described in the algorithm \ref{alg1}, a large enough $K, M$, and a large confidence interval $(1 - \alpha)$. The consistency of the coverage accuracy follows from letting $K$ and $n$ increase.

\begin{figure}
\includegraphics[width=6in]{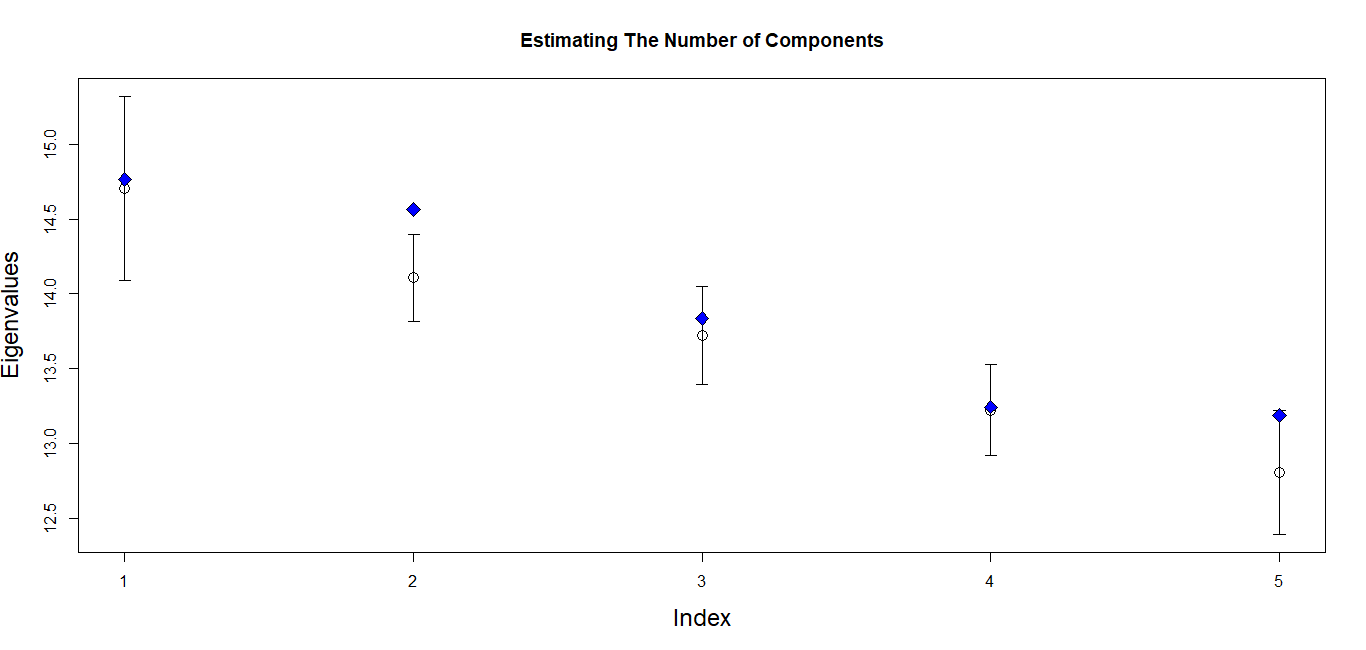}
\caption{Visual representation of the scheme to obtain the preliminary point estimate of the number of components. The depicted confidence interval is constructed using $b-1$ eigenvalues of  subsampled covariance matrices whereas square dots denote the eigenvalues of the sample covariance matrix of the lone subsampled matrix. The preliminary point estimate is one less than the first time when the square dots are not in the confidence interval.}\label{F:2}
\end{figure}


\section{Numerical Experiments}\label{sim}
We evaluated the performance of Algorithm \ref{alg1} (SS) by computing coverage probability for various simulation scenarios that varied with the true number of components $r$, the aspect ratio $\gamma_n$, and the scale of the components $\{\theta_i\}_{i=1}^r$. In particular, we vary the true number of factors $r$ over $\{0,1, 2, 3, 4, 5\}$, the aspect ratio $\gamma_n$ over $\{0.2, 0.5\}$. We choose the scale of components in an increasing manner, as shown in Table \ref{T.0}. The rationale for choosing the scale of components $\{ \theta_i \}_{i=1}^r$ in Table \ref{T.0} is to avoid any multiplicity of spiked eigenvalues while keeping the spiked eigenvalues moderately large. We use the settings in Table \ref{T.0} to test our method when the underlying data matrix was generated under a factor model and a principal component model.

\begin{table}[H]
\caption{Choice of Scaling Factor. }\label{T.0}
\centering
\begin{tabular}{| l | l |}
\hline
 True Number of Components & Scale Factor \\
\hline
 0 & NA \\
1 & \{10\} \\
2 & \{10, 15\} \\
3 & \{10, 15, 20\} \\
4 & \{10, 15, 20, 25\} \\
5 & \{10, 15, 20, 25, 30\} \\
\hline
\end{tabular}
\end{table}

Tables \ref{T.1} -\ref{T.2} compile the coverage probability for our proposed confidence intervals at range of cofidence levels over
$\{ 5\%, 10 \%, 20\%, 40\%, 60\%, 80\%, 90\%, 95\% \}$ for the simulation scenarios in Table \ref{T.0}. For computing the coverage probability, we generated the data matrix $100$ times using the factor model in (\ref{matrix.factor}) with the parameter settings in Table \ref{T.0}. Then, for each of the 100 simulations we used Algorithm \ref{alg1} 

As discussed in Section \ref{theory}, Tables \ref{T.1}-\ref{T.2} confirm that the performance of our proposed confidence interval is better for smaller and larger confidence levels, i.e., \{$0.05, 0.1$\} or \{$0.9, 0.95$\} confidence levels. Adding to this observation, Tables \ref{T.1} clearly underlines that the performance of our confidence interval for middling confidence levels is highly inacccurate at larger number of components. This effect is somewhat mitigated for larger value of $\gamma_n$, i.e., for $\gamma_n =0.5$, see Table \ref{T.2}. This
holds because for larger $\gamma_n$, the standard deviation in Theorem \ref{cov.prob.thm} is larger.
Moreover as discussed, there is a trend that the performance of our proporsed confidence intervals worsen with increasing number of components. 

\subsection{Comparison}
In the previous subsection, we studied the performance of Algorithm \ref{alg1} when the data matrix $\textbf{X}$ satisfies the assumptions \ref{a1}-\ref{a6}. However in practice, it is important to evaluate the performance of Algorithm \ref{alg1} under some
stress scenarios. Motivated by Fan, Guo, and Zheng (2020)\cite{fan2019factor}, we consider the following stress-testing scenarios.

\begin{enumerate}
\item\label{a} Scenario a) ($r=3$), $(\theta_1, \theta_2, \theta_3) = (1,1,10)$, the noise factor = $(1, 1,1 )$, $\gamma=(0.2)$,
\item\label{b} Scenario b) ($r=3$), $(\theta_1, \theta_2, \theta_3) = (10,10,1)$, the noise factor = $(1, 1,1 )$, $\gamma=(0.2)$,
\item\label{c} Scenario c) ($r=3$), $(\theta_1, \theta_2, \theta_3) = (1,1,1)$, the noise factor = $(6, 6, 6 )$, $\gamma=(0.2)$,
\item\label{d} Scenario d) ($r=3$), $(\theta_1, \theta_2, \theta_3) = (1,1,10)$, the noise factor = $(1, 1,1 )$, $\gamma=(0.5)$,
\item\label{e} Scenario e) ($r=3$), $(\theta_1, \theta_2, \theta_3) = (10,10,1)$, the noise factor = $(1, 1,1 )$, $\gamma=(0.5)$,
\item\label{f} Scenario f) ($r=3$), $(\theta_1, \theta_2, \theta_3) = (1,1,1)$, the noise factor = $(6, 6,6 )$, $\gamma=(0.5)$.
\end{enumerate}

Scenarios (\ref{a})- (\ref{f}) is used for stress testing our method and provides a good yardstick for comparing competing methods.
Scenarios (\ref{a}), (\ref{d}) consists of two low strength components and one high strength factor for different value of
$\gamma$. Scenarios (\ref{b}), (\ref{e}) consists of two high strength factors and one low strength factor for $\gamma$
varying over $(0.2, 0.5)$. Scenarios (\ref{c}), (\ref{f}) consists of three low strength factors in presence of elevated level
of random noise.

For evaluating the performance of Algorithm \ref{alg1} (SS) under the stress-scenarios, we compare the performance of Algorithm \ref{alg1} against a couple of existing methods. In particular, we compare Algorithm \ref{alg1} (SS) against BEMA, a method for computing confidence interval proposed by Ke, Ma, and Lin (2020) \cite{ke2020ci} and a non-parametric bootstrap method (NB).
Recall that BEMA relies on approximating the bulk eigenvalue using the weighted quantile of the Marchenko-Pastur distribution.
Therefore in practice, we find that the confidence intervals obtained by BEMA can be rather narrow. For NB method, we sample rows from the data matrix with replacement $M$ times to get $M$ estimates of the number of components. Averaging over $M$ times gives us a point estimate of the number of components. Repeating the above process $M$ times gives us $M$ realizations of the point estimate which we use to estimate standard deviation of the point estimate. Then, we use the normal theory of confidence intervals to get a confidence interval for the number of components.

Table \ref{T.3} compiles coverage probability of Algorithm \ref{alg1} (SS), NB, and BEMA for scenarios (\ref{a})-(\ref{f}). In Scenarios (\ref{a})-(\ref{b}), (\ref{d})-(\ref{e}) the performance of Algorithm \ref{alg1} remains stable wheras the performance of NB and BEMA 
is bad. The bad performance of NB and BEMA can be attributed to the lower strength of signals. For scenarios (\ref{c}), (\ref{f}), the performance ofAlgorithm \ref{alg1} worsens but fare better in comparison to the performance of NB and BEMA. In these scenario, NB performs overestimates the coverage probability whereas BEMA overestimates. The problem with BEMA is that it relies on approximating the bulk eigenvalues and therefore are susceptible

\begin{table}
\caption{Estimated coverage probability of confidence interval for the number of components when
the number of rows and columns of the data matrix was $(1500, 300)$.
}\label{T.1}
\centering
\begin{tabular}{| l | l | l | l | l | l | l | l | l | l |}
\hline
Model & True r & 5 \% CI   & 10 \% CI  & 20 \% CI  & 40 \% CI  &  60 \% CI  & 80 \% CI   &  90 \% CI  &  95 \% CI  \\
\hline
FA & 0 & 0.03 & 0.05 & 0.09 & 0.32 & 0.51 & 0.70 & 0.85 & 0.94 \\
FA & 1 & 0.05 & 0.07 & 0.16 & 0.30 & 0.54 & 0.77 & 0.86 & 0.93 \\
FA & 2 & 0.02 & 0.04 & 0.18 & 0.43 & 0.59 & 0.83 & 0.93 & 0.96 \\
FA & 3 & 0.07 & 0.13 & 0.21 & 0.40 & 0.66 & 0.90 & 0.95 & 0.98 \\
FA & 4 & 0.02 & 0.03 & 0.05 & 0.17& 0.31 & 0.56 & 0.8 & 0.88 \\
FA & 5 & 0.03 & 0.07 & 0.12 & 0.28 & 0.55 & 0.75 & 0.85 & 0.92 \\
\hline
PCA & 0 & 0.05 & 0.1 & 0.18 & 0.38 & 0.62 & 0.81 & 0.91 & 0.94 \\
PCA & 1 & 0.05 & 0.08 & 0.14 & 0.21 & 0.4 & 0.68 & 0.78 & 0.90 \\
PCA & 2 & 0.07 & 0.15 & 0.31 & 0.56 & 0.79 & 0.94 & 0.99 & 1 \\
PCA & 3 & 0.04 & 0.09 & 0.17 & 0.47 & 0.74 & 0.90 & 0.94 & 0.97 \\
PCA & 4 & 0.06 & 0.07 & 0.17 & 0.41 & 0.58 & 0.76 & 0.86 & 0.93 \\
PCA & 5 & 0.01 & 0.04 & 0.07 & 0.21 & 0.39 & 0.68 & 0.81 & 0.89 \\
\hline
\end{tabular}
\end{table}

\begin{table}
\caption{Estimated coverage probability of confidence interval for the number of components when
the number of rows and columns of the data matrix was $(1500, 750)$.
}\label{T.2}
\centering
\begin{tabular}{| l | l | l | l | l | l | l | l | l | l |}
\hline
Model & True r & 5 \% CI   & 10 \% CI  & 20 \% CI  & 40 \% CI  &  60 \% CI  & 80 \% CI   &  90 \% CI  &  95 \% CI  \\
\hline
FA & 0 & 0.06 & 0.11 & 0.23 & 0.43 & 0.66 & 0.91 & 0.98 & 0.99 \\
FA & 1 & 0.08 & 0.09 & 0.12 & 0.37 & 0.56 & 0.81 & 0.97 & 0.99 \\
FA & 2 & 0.06 & 0.11 & 0.2 & 0.45 & 0.67 & 0.86 & 0.95 & 0.99 \\
FA & 3 & 0.04 & 0.07 & 0.17 & 0.41 & 0.60 & 0.75 & 0.86 & 0.96 \\
FA & 4 & 0.06 & 0.11 & 0.30 & 0.54 & 0.76 & 0.94 & 0.97 & 0.97 \\
FA & 5 & 0.03 & 0.07 & 0.16 & 0.36 & 0.53 & 0.7 & 0.84 & 0.93 \\
\hline
PCA & 0 & 0.02 & 0.06 & 0.12 & 0.42 & 0.62 & 0.89 & 0.96 & 0.98 \\
PCA & 1 & 0.04 & 0.10 & 0.20 & 0.35 & 0.57 & 0.81 & 0.92 & 0.96 \\
PCA & 2 & 0.02 & 0.06 & 0.14 & 0.34 & 0.62 & 0.90 & 0.97 & 0.99 \\
PCA & 3 & 0.02 & 0.07 & 0.11 & 0.22 & 0.40 & 0.74 & 0.88 & 0.95 \\
PCA & 4 & 0.03 & 0.06 & 0.11 & 0.29 & 0.54 & 0.74 & 0.85 & 0.92 \\
PCA & 5 & 0.04 & 0.09 & 0.18 & 0.36 & 0.58 & 0.77 & 0.85 & 0.94 \\
\hline
\end{tabular}
\end{table}

\begin{table}
\caption{Estimated coverage probability of confidence interval for the number of components.
}\label{T.3}
\centering
\begin{tabular}{| l | l | l | l | l | l | l | l | l | l | }
\hline
Method & Sc & 5 \% CI   & 10 \% CI  & 20 \% CI  & 40 \% CI  &  60 \% CI  & 80 \% CI   &  90 \% CI  &  95 \% CI  \\
\hline
SS & a &  0.05 & 0.11 & 0.22& 0.38 & 0.63 & 0.78 & 0.86 & 0.92 \\
NB & a &  0 & 0 & 0 & 0 & 0 & 0 & 0 & 0 \\
BEMA & a & 0 & 0 & 0 & 0 & 0 & 0 & 0 & 0 \\
\hline
SS & b &  0.05 & 0.14 & 0.22 & 0.50 & 0.71 & 0.92 & 0.98 & 0.99 \\
NB & b &  0 & 0 & 0 & 0 & 0 & 0 & 0 & 0 \\
BEMA & b & 0 & 0 & 0 & 0 & 0 & 0 & 0 & 0 \\
\hline
SS & c  & 0.06 & 0.16 & 0.28 & 0.53 & 0.76 & 0.96 & 0.98 & 1 \\
NB & c  & 0.5 & 0.73 & 0.93 & 1 & 1 & 1 & 1 & 1 \\
BEMA & c  & 0 & 0 & 0 & 0 & 0 & 0 & 0 & 0 \\
\hline
SS & d  & 0.06 & 0.18 & 0.31 & 0.50 & 0.77 & 0.92 & 1 & 1 \\
NB & d  & 0 & 0 & 0 & 0 & 0 & 0 & 0 & 0 \\
BEMA & d  & 0 & 0 & 0 & 0 & 0 & 0 & 0 & 0 \\
\hline
SS & e & 0.05 & 0.17 & 0.29 & 0.61 & 0.79 & 0.93 & 0.98 & 0.99 \\
NB & e  & 0 & 0 & 0 & 0 & 0 & 0 & 0 & 0 \\
BEMA & e  & 0 & 0 & 0 & 0 & 0 & 0 & 0 & 0 \\
\hline
SS & f  & 0.07 & 0.13 & 0.28 & 0.56 & 0.78 & 0.89 & 0.96 & 1 \\
NB & f  & 0 & 0.02 & 0.09 & 0.56 & 0.93 & 1 & 1 & 1 \\
BEMA & f  & 0 & 0 & 0 & 0 & 0 & 0 & 0 & 0 \\
\hline
\end{tabular}
\end{table}

\begin{figure}
\includegraphics[width=5.5in, height=4.5in]{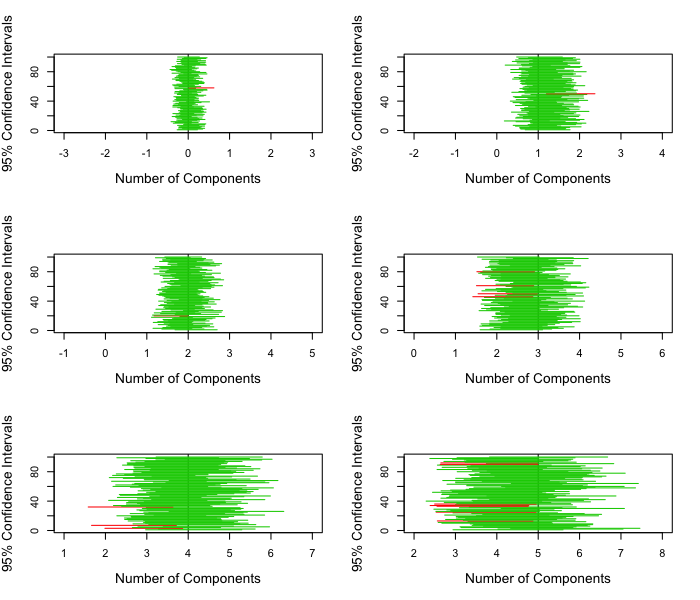}
\caption{Visual representation of the performance of the proposed confidence interval at 95\% confidence level. We plot  hundred 95\% confidence intervals (of the number of components) for dataset of dimension $1500 \times 750$ generated using FA 
for the true number of components $r$ varying over $\{0,1,2,3,4,5\}$. Each confidence is color coded with green (red) denoting that the confidence interval covers (doesn't cover) the true number of components. Every scenario has different number of the true number of components with vertical bar indicating the true number of components.}
\end{figure}

\begin{figure}
\includegraphics[width=5.5in, height=4.5in]{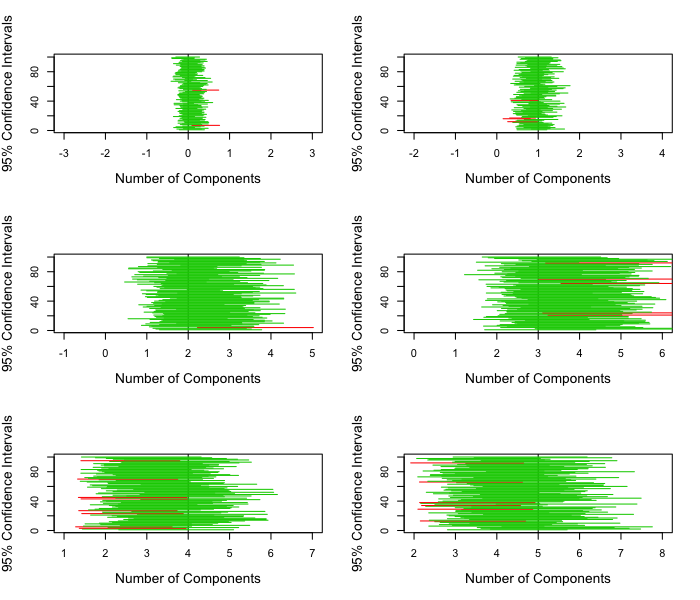}
\caption{Visual representation of the performance of the proposed confidence interval at 95\% confidence level. We plot  hundred 95\% confidence intervals (of the number of components) for dataset of dimension $1500 \times 750$ generated using PCA 
for the true number of components $r$ varying over $\{0,1,2,3,4,5\}$. Each confidence is color coded with green (red) denoting that the confidence interval covers (doesn't cover) the true number of components. Every scenario has different number of the true number of components with vertical bar indicating the true number of components.}
\end{figure}



\section{Real Data Analysis}\label{data}
For real data analysis, we consider the single nucleotide polymorphism (SNP) genotyping dataset from the Human Genome Diversity Project (HGDP) dataset (\cite{cannhgdp2002},\cite{lihgdp2008}). The dataset was collected for studying the genetic variation across the human populations around the world. The SNP genotyping dataset is collected for $1043$ samples from $51$ representing populations from Africa, Europe, Asia, Oceania, and the Americas. The dataset can be downloaded from \emph{https://hagsc.org/hgdp/files.html}.

Following \cite{dobriban2018dpa}, we study the SNP genotyping data for one of the smallest human chromosome $22$. The SNP genotyping dataset on chromosome $22$ comprises of $1043$ samples and $9743$ SNPs. Since we did not have access to the reference allele, we used the allele frequency across $1043$ samples to determine the minor allele at every RSid site of the chromosome. Using the minor allele determination, we converted the SNP genotyping dataset into a numerical matrix of dimension $n \times p$, where $X_{ij} \in \{0, 1, 2\}$ is the number of minor allele of SNP $j$ (at a particular RSid location on the chromosome $22$) in the $i^{th}$ individual. For analyzing the dataset, we replaced any missing values with $0$. Moreover, we centered and scale the remaining dataset SNP-wise. As we discuss, we notice that some results are somewhat different from \cite{dobriban2018dpa}. This is because they had access to the reference allele or a different method for determining minor allele.

We used Likelihood-based method, DPA, and DDPA+ for estimating the number of components. Then, for each of these estimate of the
number of components, we estimate a $99\%$ confidence interval for the number of components. Surprisingly, we notice that the $99 \%$ confidence interval for Likelihood-based method and DPA does not contain the respective estimated number of components. This phenomena can be explained as follows. In particular, this is because both Likelihood-based method and DPA estimated 

Recall that the algorithm \ref{alg1} is constructed using the eigenvalues of the subsampled covariance matrices that have relatively weaker signal strength compared to the original data matrix. Moreover, the empirical histogram of the number of components in algorithm \ref{alg1} is obtained after averaging over conditonally independent subsamples that also brings down the
empirical histogram of the number of components. However, when the estimated number of components is not large, as in the case of DDPA+. Based on the above discussion, it is fair to conclude that the estimated confidence interval tend to be more robust whereas individual methods for estimating the number of components.

\begin{table}
\caption{99 \% Confidence Interval for the Number of Components of SNP genotyping dataset with different estimate of the number of components.}\label{T.3}
\centering
\begin{tabular}{| l | l | l | l | l | l | l | l | l |}
\hline
Estimate & \multicolumn{2}{c}{Likelihood-Based} & \multicolumn{2}{c}{DPA} & \multicolumn{2}{c}{DDPA+} & \multicolumn{2}{c}{BEMA} \vline  \\
\hline
 & Lower & Upper & Lower & Upper & Lower & Upper & Lower & Upper \\
\hline
99 \% CI & 9 & 14 & 24 & 30 & 8 & 15 & 117 & 122 \\
\hline
\end{tabular}
\end{table}

\begin{table}
\caption{99 \% Confidence Interval for the Number of Components of SNP genotyping dataset with different estimate of the number of components.}\label{T.3}
\centering
\begin{tabular}{| l | l | l | l |    }
\hline
 Likelihood-Based & DPA & DDPA+ & BEMA  \\
\hline
 178 & 212 &13 & 120\\
\hline
\end{tabular}
\end{table}

Since the truth about the true number of components is not known about this dataset, we cannot simply validate our $99 \%$ confidence interval estimate of the number of components. Nonetheless, it is prudent to provide some heuristic validation for our
estimated $99 \%$ confidence interval. Motivated by \cite{perryowen2010} and \cite{dobriban2018dpa}, we provide graphical
validation for our estimated $99 \%$ confidence interval. Figures show

\begin{figure}
\centering
\includegraphics[width=6in]{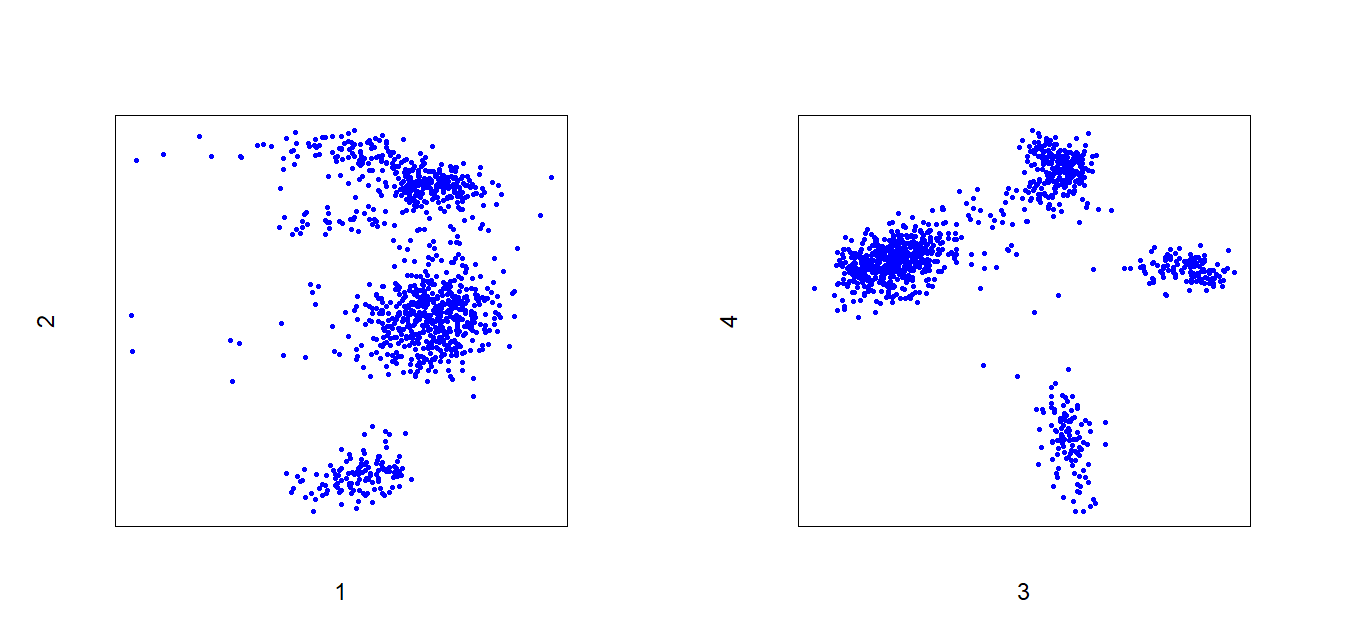}
\includegraphics[width=6in]{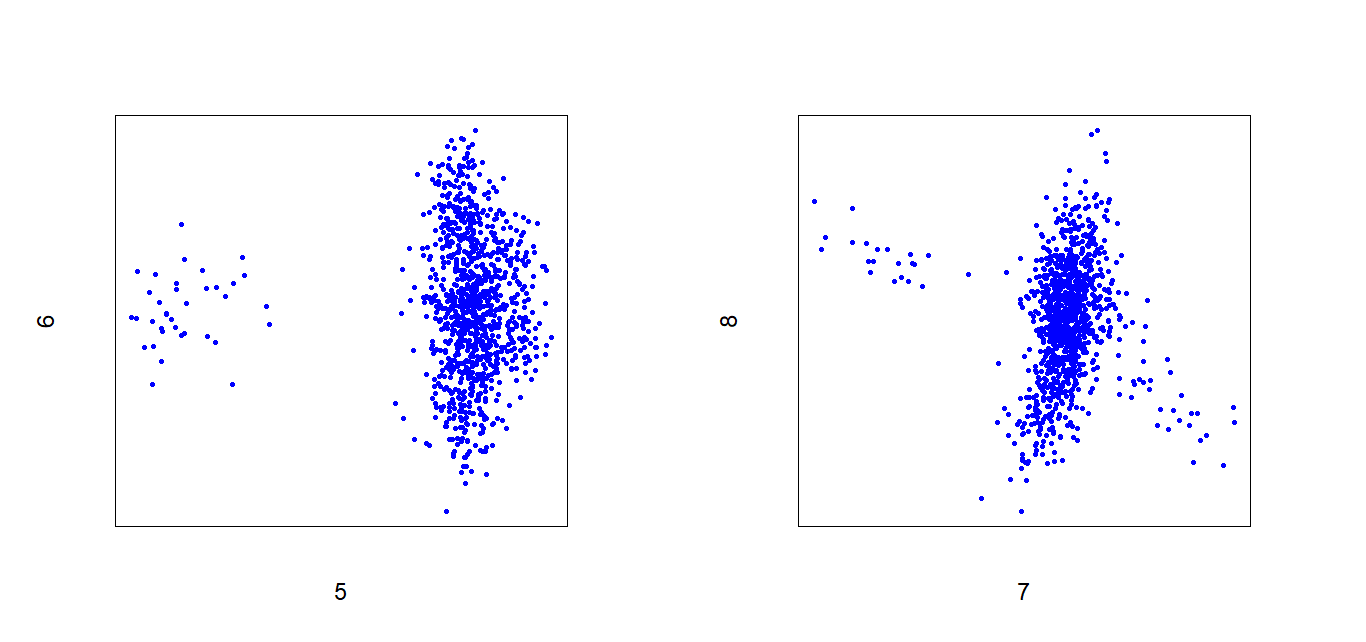}
\caption{Visual representation of the left singular vectors $1-8$ of HGDP data}\label{F:0}
\end{figure}

\begin{figure}
\centering
\includegraphics[width=6in]{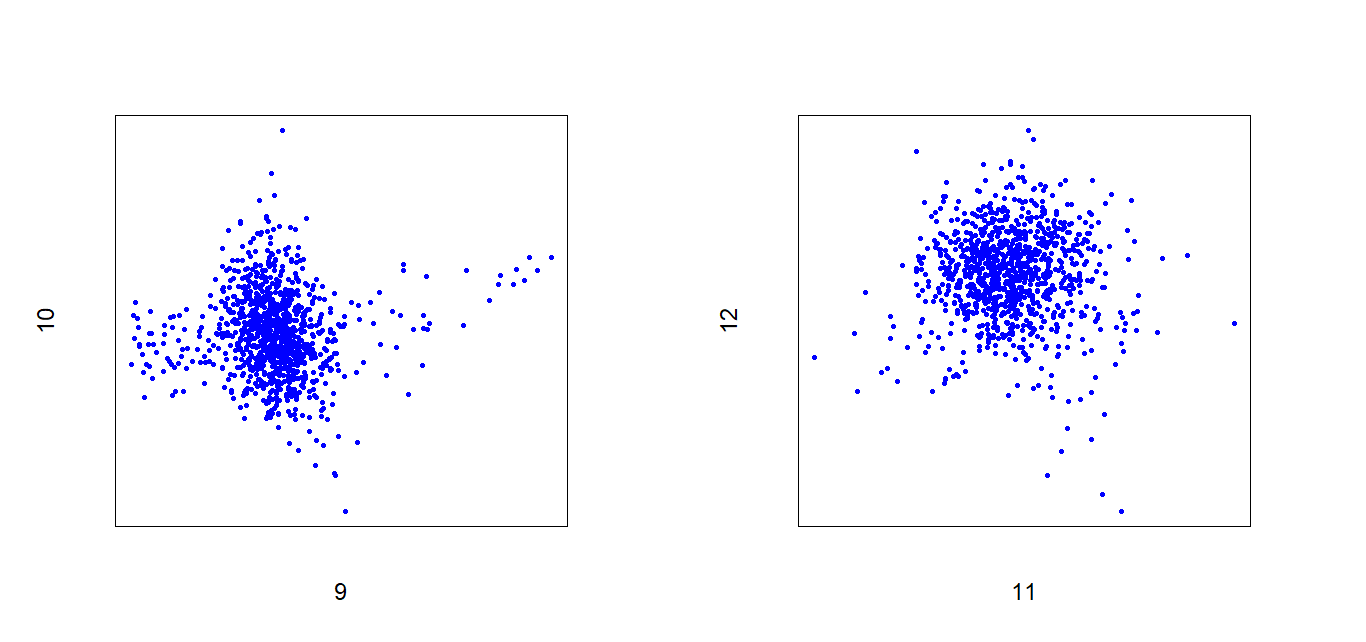}
\includegraphics[width=6in]{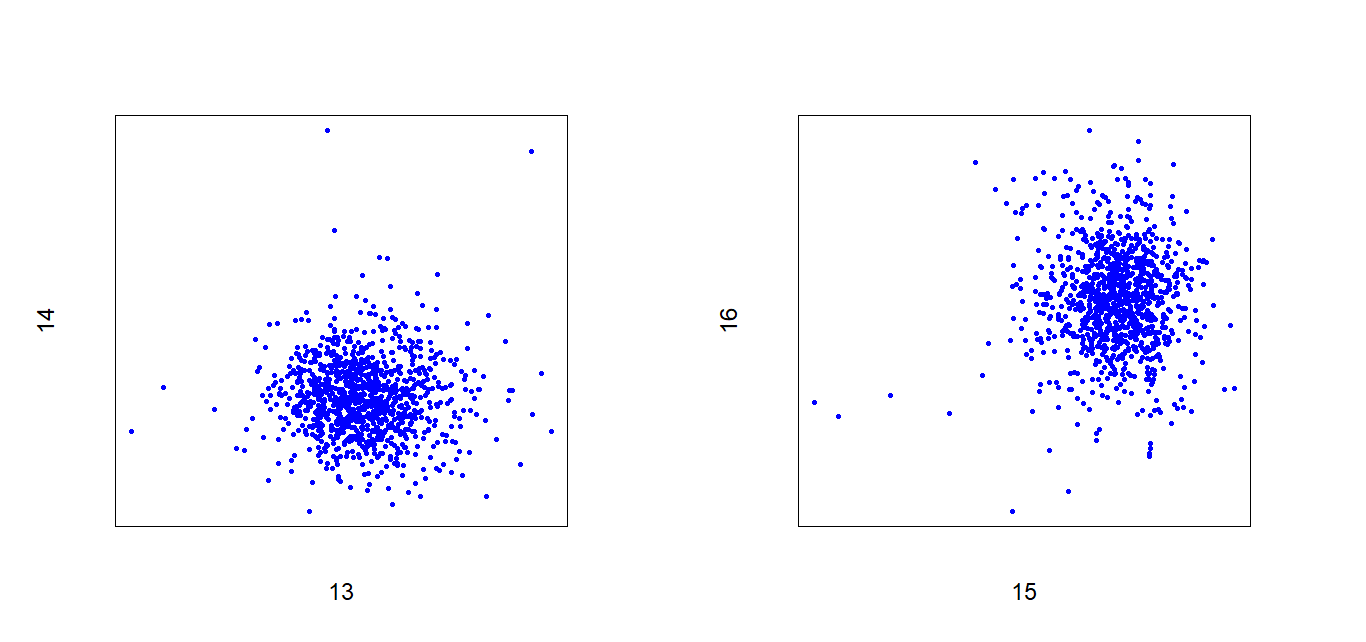}
\caption{Visual representation of the left singular vectors $9-16$ of HGDP data}\label{F:0}
\end{figure}

\section{Discussion}\label{disc}
We proposed a subsampling algorithm for obtaining confidence intervals for the number of components in FA and PCA.
We also showed that the coverage accuracy of the proposed confidence interval is good when the confidence level is high.
The above theoretical result is based on obtaining the \emph{Edgeworth} expansion of the spiked eigenvalues of the
sample covariance matrix. Although Yang and Johnstone (2018) \cite{johnstone2018pca} obtained the \emph{Edgeworth}
expansion for the spiked eigenvalues of the sample covariance matrix under benign conditions for PCA but our assumptions
in the case of FA is comparatively demanding. Therefore, we believe that it may be possition to improve upon these conditions.

\bibliography{reference}

\section{Acknowledgments}
We are grateful to Dr. Edgar Dobriban at the Department of Statistics and Data Science at the University of Pennsylvania for
introducing us to the above problem. We are also thankful to Dr. Edgar Dobriban and Dr. Fan Yang (also at
the Depatment of Statistics and Data Science at the University of Pennsylvania) for rich discussion on this topic.

\end{document}